\documentclass[a4paper,10pt]{article}
\usepackage[ansinew]{inputenc}
\usepackage{amsthm, amsmath}
\usepackage{graphicx}
\usepackage{subcaption}
\usepackage{algpseudocode,algorithm,algorithmicx}
\usepackage{color}
\usepackage{multirow}
\usepackage{lineno}
\usepackage{cite}
\usepackage{fullpage}

\newtheorem{theorem}{Theorem}[section]

\newtheorem{conjecture}[theorem]{Conjecture}

\algrenewcommand\algorithmicrequire{\textbf{Input:}}
\algrenewcommand\algorithmicensure{\textbf{Output:}}

\newcommand\blfootnote[1]{%
  \begingroup
  \renewcommand\thefootnote{}\footnote{#1}%
  \addtocounter{footnote}{-1}%
  \endgroup
}

\date{}

\title{Plane augmentation of plane graphs to meet parity constraints}
\author{J.C. Catana\thanks{{\tt j.catanas@uxmcc2.iimas.unam.mx}. Universidad Nacional Aut\'onoma de M\'exico, Mexico.}
\and A. Garc\'\i a\thanks{{\tt olaverri@unizar.es}. IUMA, Universidad de Zaragoza, Spain.}
\and J. Tejel\thanks{{\tt jtejel@unizar.es}. IUMA, Universidad de Zaragoza, Spain.}
\and J. Urrutia \thanks{{\tt urrutia@matem.unam.mx}. Universidad Nacional Aut\'onoma de M\'exico, Mexico.}}

\begin{document}

\maketitle

\blfootnote{\begin{minipage}[l]{0.3\textwidth} \includegraphics[scale=0.13]{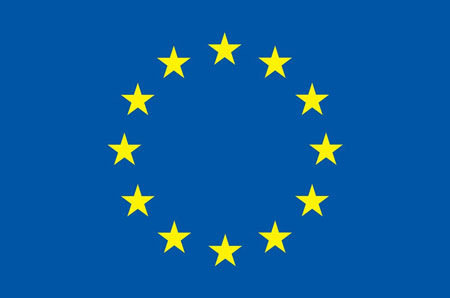} \end{minipage}  \hspace{-2cm} \begin{minipage}[l][1cm]{0.7\textwidth}
 	  This project has received funding from the European Union's Horizon 2020 research and innovation programme under the Marie Sk\l{}odowska-Curie grant agreement No 734922.
 	\end{minipage}}

\begin{abstract}
A plane topological graph $G=(V,E)$ is a graph drawn in the plane whose vertices are points in the plane and whose edges are simple curves that do not intersect, except at their endpoints. Given a plane topological graph $G=(V,E)$ and a set $C_G$ of parity constraints, in which every vertex has assigned a parity constraint on its degree, either even or odd, we say that $G$ is \emph{topologically augmentable} to meet $C_G$ if there exits a plane topological graph $H$ on the same set of vertices, such that $G$ and $H$ are edge-disjoint and their union is a plane topological graph that meets all parity constraints.

In this paper, we prove that the problem of deciding if a plane topological graph is topologically augmentable to meet parity constraints is $\mathcal{NP}$-complete, even if the set of vertices that must change their parities is $V$ or the set of vertices with odd degree. In particular, deciding if a plane topological graph can be augmented to a Eulerian plane topological graph is $\mathcal{NP}$-complete. Analogous complexity results are obtained, when the augmentation must be done by a plane topological perfect matching between the vertices not meeting their parities.

We extend these hardness results to planar graphs, when the augmented graph must be planar, and to plane geometric graphs (plane topological graphs whose edges are straight-line segments). In addition, when it is required that the augmentation is made by a plane geometric perfect matching between the vertices not meeting their parities, we also prove that this augmentation problem is $\mathcal{NP}$-complete for plane geometric trees and paths.

For the particular family of maximal outerplane graphs,
we characterize maximal outerplane graphs that are topological augmentable to satisfy a set of parity constraints. We also provide a polynomial time algorithm that decides if a maximal outerplane graph is topologically augmentable to meet parity constraints, and if so, produces a set of edges with minimum cardinality.
\end{abstract}

	
	\begin{figure}[ht!]
	\centering
		\begin{subfigure}[t]{0.3\textwidth}
			\includegraphics[width=\linewidth,page=1]{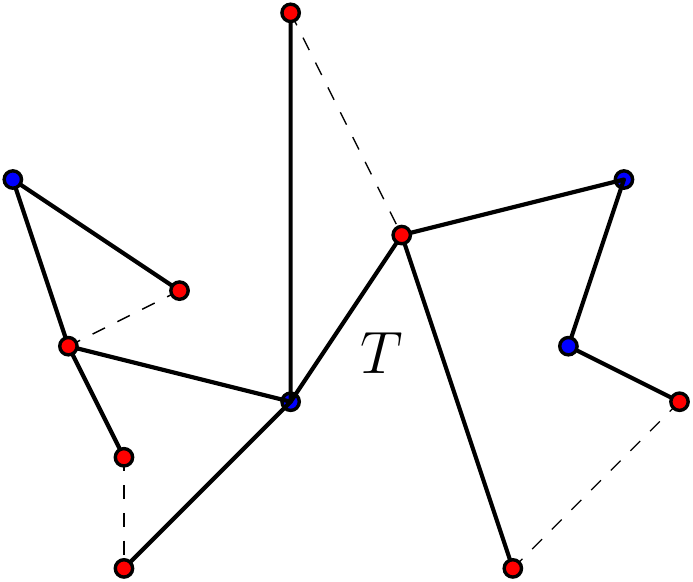}
			\captionsetup{width=.95\linewidth}
			\caption{}
			\label{fig:example1}
		\end{subfigure}~~~~~  %
		\begin{subfigure}[t]{0.3\textwidth}
			\includegraphics[width=\linewidth,page=2]{Examples.pdf}
			\captionsetup{width=.95\linewidth}
			\caption{}
			\label{fig:example2}
		\end{subfigure}~~~%
		\begin{subfigure}[t]{0.34\textwidth}
			\includegraphics[width=\linewidth,page=3]{Examples.pdf}
			\captionsetup{width=.95\linewidth}
			\caption{}
			\label{fig:example3}
		\end{subfigure}%
		\caption{Plane geometric augmentation problems to meet a set of parity constraints. Vertices in $R$ are depicted as red points and the rest as blue points. We use thick lines for the given plane geometric graph and dashed lines for the added edges.}\label{fig:exmaples}
	\end{figure}

\section{Introduction}

An \emph{embedding} of a graph $G=(V,E)$ on the plane is a drawing of $G$ on the plane in which its vertices are represented by points and its edges by simple curves joining pairs of adjacent vertices of $G$. We shall refer to these embeddings as \emph{topologial graphs}.  A \emph{geometric graph} is a topological graph whose edges are represented by straight-line segments. A topological graph is \emph{plane} if the curves representing its edges do not intersect, except at their endpoints. Two plane topological graphs $G=(V,E)$ and $H=(V,E')$ on the same set of vertices are called \emph{compatible} if  their union is a plane topological graph.

Let $G$ be a plane topological graph with $n$ vertices $v_1, \ldots, v_{n}$. Suppose that each vertex $v_i$ of $G$ has been assigned a parity constraint $c_i$, that is, $c_i$ is either \emph{even} or \emph{odd}. The set $C_G=\{c_1, c_2, ..., c_{n}\}$ will be called the set of parity constraints of $G$. The \emph{plane topological augmentation problem} of $G$ to meet $C_G=\{c_0, c_1, ..., c_{n-1}\}$ is that of finding a plane topological graph $H$ such that $G$ and $H$ are compatible and edge-disjoint, and the degree of every vertex in $G \cup H$ meets its parity constraint in $C_G$. If there exists such a graph $H$, then we say that $G$ is \emph{topologically augmentable} to meet $C_G$.

Analogously, given a plane geometric graph $G$ and a set $C_G$ of parity constraints, the \emph{plane geometric augmentation problem} of $G$ to meet $C_G$ is that of finding a plane geometric graph $H$ such that $G$ and $H$ are compatible and edge-disjoint, and the degree of every vertex in $G \cup H$ meets its parity constraint in $C_G$. If there exists such a graph $H$, then we say that $G$ is \emph{geometrically augmentable} to meet $C_G$.

Hereafter we denote by $R(C_G)$ (or simply $R$ if there is no ambiguity) the set of vertices in $G$ not satisfying their parity constraints in $C_G$, that is the set of vertices in $G$ that must change their parities. For simplicity in further reasoning, a vertex $v_i \in R(C_G)$ is called a \emph{red vertex}, otherwise it is called a \emph{blue vertex}. We say that an edge of $G$ is a \emph{red edge} (respectively \emph{blue edge}), if both of its endpoints are red vertices (respectively blue). An edge of $G$ is called a \emph{red-blue edge} if its endpoints have different colors.

Figure~\ref{fig:exmaples} shows some examples of plane geometric augmentation problems to meet a set of parity constraints. In Figure~\ref{fig:example1}, a plane geometric tree $T$ is augmented to a Eulerian plane geometric graph (all vertices have even degree), by adding a plane geometric perfect matching between the vertices with odd degree in $T$. Figure~\ref{fig:example2} shows a plane geometric path $P$ such that there is only one way of changing the parities of all vertices except $v$. In Figure~\ref{fig:example3},
the parities of the ten red vertices cannot all be changed simultaneously. These vertices define two empty convex pentagons $P_1$ and $P_2$, and the parity of only two vertices in each pentagon can be changed using diagonals. This construction can be generalized by adding more empty convex pentagons, and the resulting $n$-vertex graph is a plane geometric graph such that the parity constraints of at most $\frac{2n}{5}$ vertices can be satisfied in any plane augmentation.

Note that if a plane topological graph $G$ can be augmented to a plane topological graph $G'=G\cup H$ to satisfy a set $C_G$ of parity constraints, the degree in $H$ of all red vertices in $G$ is odd, while the degree of all blue vertices is even. It follows that the number of red vertices of $G$ must be even since any graph has an even number of odd vertices. Moreover, if $H$ has as few edges as possible, then $H$ is a forest. Otherwise, if a connected component of $H$ contains a cycle $C$, all of the edges of $C$ could be removed from $H$ without changing the parities of the vertices in $G \cup H$.
Therefore, the number of edges in $H$ is at least $|R(C_G)|/2$. If $H$ has $|R(C_G)|/2$ edges, then $H$ is a perfect matching between the vertices in $R(C_G)$.

\subsubsection*{Related previous work}

Graph augmentation is a family of problems in which one would like to add new edges,  ideally as few as possible, to a given graph, in such a way, that some desired property is  achieved. Connectivity augmentation is one of the most studied augmentation problems, due to a wide range of applications in designing reliable networks and database systems.

Connectivity augmentation in abstract graphs was first studied by Eswaran and Tarjan in 1976 (and independently by Plesn\'ik at the same time~\cite{plesenik1976minimum}). They obtained two polynomial time algorithms to optimally augment a graph to a $2$-connected or a $2$-edge-connected graph~\cite{eswaran1976augmentation}. In the same paper, Eswaran and Tarjan also studied the problem of augmenting weighted oriented graphs to achieve strong connectivity, proving that this problem is $\mathcal{NP}$-complete.
Since then, numerous papers about this topic have been published, studying efficient algorithms to increase the vertex- and edge-connectivity (see for example~\cite{vegh2011}) or giving approximation algorithms for hard versions of the problem (see for example~\cite{Kortsarz2015}). The reader can consult~\cite{frank2011} for a broad overview on these connectivity augmentation problems.

For planar graphs, Kant and Bodlaender~\cite {Kant1991aa} proved that finding a minimal set of edges that makes a planar connected graph planar and 2-connected is
$\mathcal{NP}$-complete. Later, Rutter and Wolff~\cite{rutter2008augmenting} extended this result by showing the $\mathcal{NP}$-completeness of the same problem replacing 2-connected by 2-edge-connected. For plane geometric graphs, these problems remain hard even if restricted to trees~\cite{rutter2008augmenting}. The reader is referred  to~\cite{abellanas2008augmenting, akitaya2019, Al-Jubeh2011, Garcia2015, kranakis2012, rutter2008augmenting, toth2012connectivity} for different results about connectivity augmentation problems in plane geometric graphs, and the survey~\cite{hurtado2013plane} for more details and related topics. Recent research on compatible plane graphs can be found in~\cite{aichholzer2018, aicholzer2009, aichholzer2011, garcia2014, Ishaque2013}.

Augmentation to meet parity constraints has also been studied in the literature, as a weaker version of the more general classic problem of augmenting a graph to meet a given degree sequence. Parity constraints characterize some types of graphs such as Eulerian graphs or maximal planar graphs that are 3-colorable (in both cases, connected graphs where all vertices have even degree). For abstract graphs, Dabrowski et al.~\cite{dabrowski2016} presented a polynomial algorithm to solve the augmentation problem to meet a set of parity constraints. This result extends the polynomial algorithm for the particular case of augmenting a graph to Eulerian~\cite{boesch1977}. Dabrowski et al.~\cite{dabrowski2016} also provided a polynomial algorithm to obtain a graph meeting the parity constraints, when addition and removal of edges are allowed. More results on editing graphs to make them Eulerian are available~\cite{cai2011, Cygan2014, dabrowski2016, dorn2013}. The relevance of Eulerian graphs relies on a wide range of applications to some problems in different areas such as route inspection problems, DNA fragmentation problems, scheduling problems, or the design of CMOS VLSI circuits. The reader can also see~\cite{Dong2018, Zhang2012} for other related results.

The geometric setting of the problem was also been studied~\cite{aicholzer2014, alvarez2015parity}. Given a point set $S$ in the plane and a set $C_S$ of parity constraints on the points of $S$, Aichholzer et al.~\cite{aicholzer2014} proved that it is always possible to build up a plane tree and a $2$-connected outerplanar graph satisfying all parity constraints, and a pointed pseudotriangulation that satisfy all but at most three parity constraints. For triangulations, they showed examples in which a linear number of parity constraints cannot be satisfied, and they provided triangulations satisfying about $\frac{2n}{3}$ of the parity constraints. Later, Alvarez~\cite{alvarez2015parity} showed that it is always possible to build up an even triangulation (a triangulation having all its vertices with even degree) on top of a point set $S \cup S'$, where $S$ is a $n$-point set in general position and  $S'$ a set of Steiner points of at most $\frac{n}{3} + c$ points, with $c$ being a positive integer constant.

\subsubsection*{Our results}

In this paper, we study the plane augmentation problem to meet parity constraints. In Section~\ref{sec:plane}, we focus on the topological plane augmentation problem. Given a plane topological graph $G=(V,E)$ and a set $C_G$ of parity constraints, we show that deciding if $G$ is topologically augmentable to meet $C_G$ is $\mathcal{NP}$-complete, even if $R(C_G)$ coincides with $V$ or is the set of vertices with odd degree in $G$. In particular, deciding if a plane topological graph can be augmented to a Eulerian plane topological graph is $\mathcal{NP}$-complete. The same complexity results are obtained, when the augmentation must be done by a plane topological perfect matching between the red vertices. Hence, deciding if $G$ admits a plane topological perfect matching, compatible with $G$ and edge-disjoint, is also $\mathcal{NP}$-complete. As a side result, we also extend these $\mathcal{NP}$-completeness results to planar graphs.

\begin{table}[!hbt]
\begin{center}
\begin{tabular}{|l|c|c|c|c|c|c|}
  \hline
   & \multicolumn{6}{|c|}{Topological augmentation} \\
     \cline{2-7}
   & \multicolumn{3}{|c|}{Decision}& \multicolumn{3}{|c|}{Matching} \\
   \cline{2-7}
   & Any $R$ & $R=V$ & Eulerian & Any $R$ & $R=V$ & Eulerian \\
   \hline
  Plane topological graph & {\bf NP} & {\bf NP} & {\bf NP} & {\bf NP} & {\bf NP}  & {\bf NP} \\
  Maximal outerplane graph & {\bf P} & {\bf P} & {\bf P}  & {\bf P} & {\bf P} & {\bf P} \\
  \hline
      \multicolumn{7}{|c|}{ } \\
\hline
          & \multicolumn{6}{|c|}{Geometric augmentation} \\
     \cline{2-7}
   & \multicolumn{3}{|c|}{Decision}& \multicolumn{3}{|c|}{Matching} \\
   \cline{2-7}
   & Any $R$ & $R=V$ & Eulerian & Any $R$ & $R=V$ & Eulerian \\
   \hline
  Plane geometric graph & {\bf NP} & {\bf NP} & {\bf NP} & {\bf NP} & {\bf NP} & {\bf NP} \\
  Plane geometric tree & {\bf ?} & {\bf ?} & {\bf ?} & {\bf NP} & {\bf NP} & {\bf NP}  \\
  Plane geometric path & {\bf ?} & {\bf ?} & {\bf ?} & {\bf NP} & {\bf NP} & {\bf P}  \\
  \hline
\end{tabular}
\caption{Summary of our results for the plane augmentation problem to meet a set of parity constraints. We distinguish if the augmentation is topological or geometric and if the input graph $G(V,E)$ is augmentable by an arbitrary graph (Decision) or by a plane perfect matching between the red vertices (Matching). We also distinguish the cases that the set $R$ of red vertices is arbitrary, coincides with $V$ or is the set of odd vertices in $G$ (Eulerian).}\label{tabla2}
\end{center}
\end{table}

Section~\ref{sec:geometric} is devoted to the geometric augmentation problem to meet parity constraints. As a consequence of the results shown in Section~\ref{sec:plane}, we can prove the $\mathcal{NP}$-completeness of all the problems previously described, when the input graph and the augmenting graph are plane and geometric. In addition, when it is required that the geometric augmentation is made by a plane geometric perfect matching between the red vertices, we also prove that deciding if a plane geometric tree (or path) is geometrically augmentable to meet parity constraints is $\mathcal{NP}$-complete. In particular, the two following problems are also $\mathcal{NP}$-complete: Deciding if a plane geometric tree (or path) admits a plane geometric perfect matching, compatible and edge-disjoint, and deciding if a plane geometric tree is geometrically augmentable to Eulerian by a plane geometric perfect matching between the vertices with odd degree.

Due to the hardness of the plane topological augmentation problem in general, in Section~\ref{sec:mops} we address this problem for a particular family of plane topological graphs, the family of maximal outerplane graphs. For this family, we characterize maximal outerplane graphs that are topological augmentable to satisfy all parity constraints, and we provide an $\mathcal{O}(n^3)$ time algorithm to find a minimum plane topological augmentation, if it exists. Given a maximal outerplane graph, we also show that there is always a plane topological matching such that the union of the maximal outerplane graph and the matching satisfies all but at most four parity constraints.
Table~\ref{tabla2} summarizes our results for the plane augmentation problem to meet parity constraints.

\section{Plane topological augmentation problems}\label{sec:plane}

In this section we address the plane topological augmentation problem to meet a set of parity constraints, that is, given a plane topological graph $G=(V,E)$ and a set $C_G$ of parity constraints, we look for a plane topological graph $H$ such that $G$ and $H$ are compatible and edge-disjoint, and all parity constrains are satisfied in $G\cup H$. We show the $\mathcal{NP}$-completeness of several variants of this problem and we extend these results to planar graphs.

We first recall the Planar 3-SAT problem. This problem will play a key role in the forthcoming theorems. Let $\Phi$ be a Boolean formula in conjunctive normal form and let $F_{\Phi}$ be the incidence graph of $\Phi$, that is, the bipartite graph whose vertices are the variables and the clauses of $\Phi$ and an edge connects a variable and a clause if and only if the variable (negated or unnegated) occurs in the clause. A formula is a 3-SAT formula if every clause contains at most three literals, where a literal is either a variable (called positive literal) or the negation of a variable (called negative literal). The Planar 3-SAT problem asks whether a given 3-SAT formula $\Phi$ is satisfiable, assuming that  $F_{\Phi}$ is planar. This problem has been shown to be $\mathcal{NP}$-complete~\cite{lichtenstein1982}.

Given a formula $\Phi$, suppose that $F_{\Phi}$ is planar and $F$ is a plane embedding of $F_{\Phi}$. In this case, observe that if $l_1, l_2, \ldots , l_k$ is the set of literals  of a variable $x$ in $\Phi$, then the clockwise cyclic order around $x$ of the clauses adjacent to $x$ in $F$ also implies an order for the literals of $x$, that is, two literals $l_i$ and $l_{i+1}$ are consecutive if the clauses in $F$ where they occur are consecutive in the cyclic order around $x$. In the rest of this section, $L_x=\{l_1, l_2, \ldots , l_k\}$ will denote the ordered list of literals of a variable $x$, according to a plane embedding $F$. Figure~\ref{fig:3sat} shows a plane embedding of $F_{\Phi}$ for the 3-SAT formula $\Phi = (\bar{x}_1 \vee x_2 \vee \bar{x}_3) \wedge(x_1 \vee \bar{x}_2 \vee \bar{x}_3) \wedge (x_1 \vee \bar{x}_3 \vee \bar{x}_4)$, where $\bar{x}$ means the negation of $x$. The (cyclic) ordered list of literals of $x_1$ is $L_{x_1}=\{x_1, \bar{x}_1, x_1\}$.

We next show the $\mathcal{NP}$-completeness of the plane topological augmentation problem.

\begin{figure}[ht!]
			\centering
			\includegraphics[width=.4\linewidth]{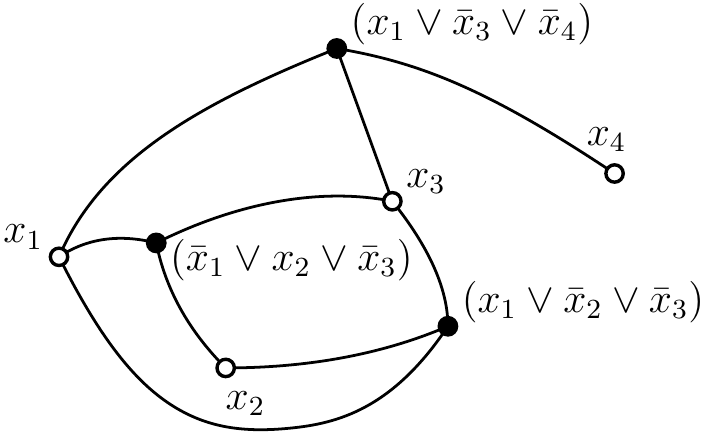}
		\caption{A plane embedding of $F_{\Phi}$ for the formula $\Phi = (\bar{x}_1 \vee x_2 \vee \bar{x}_3) \wedge(x_1 \vee \bar{x}_2 \vee \bar{x}_3) \wedge (x_1 \vee \bar{x}_3 \vee \bar{x}_4)$.}
		\label{fig:3sat}
\end{figure}

\begin{theorem}
	Let $G=(V,E)$ be a plane topological graph, and let $C_G$ be a set of parity constraints. Then, the problem of deciding if $G$ is topologically augmentable to meet $C_G$
is $\mathcal{NP}$-complete.
	\label{teo:np_topo_planes}
\end{theorem}
\begin{proof}
The plane topological augmentation problem is trivially in $\mathcal{NP}$ since, given a set of edges $E'$, verifying if $G'=(V,E \cup E')$ is a plane topological graph meeting $C_G$ takes polynomial time.
	
The proof is based on a reduction from Planar 3-SAT. Given a 3-SAT formula $\Phi$ and a plane embedding $F$ of $F_{\Phi}$, we build in polynomial time a plane topological graph $G_{\Phi}$ and a set $C_{G_{\Phi}}$ of parity constraints such that  $G_{\Phi}$ is topologically augmentable to meet $C_G$ if and only if $\Phi$ is satisfiable. The graph $G_{\Phi}$ is built from $F$, replacing  clause vertices by clause gadgets, variable vertices by variable gadgets and edges by pairs of wire gadgets. Figure~\ref{fig:3satphi} illustrates $G_{\Phi}$ for the formula $\Phi = (\bar{x}_1 \vee x_2 \vee \bar{x}_3) \wedge(x_1 \vee \bar{x}_2 \vee \bar{x}_3) \wedge (x_1 \vee \bar{x}_3 \vee \bar{x}_4)$; the plane embedding of $F_{\Phi}$ shown in Figure~\ref{fig:3sat}.

Without loss of generality, we can assume that every clause in $\Phi$ consists of three literals. The \emph{basic gadget} in our reduction (see Figure~\ref{fig:basic_gadgets}) consists of a plane topological graph that has only two possible plane topological augmentations to change the parities of all red interior vertices. These two augmentations are plane matchings that match all red vertices of the gadget except two. The \emph{negative} augmentation (depicted in the figures with red dashed edges) does not match the exterior vertices $v_r$ and $v_l$ (see Figure~\ref{fig:basic_gadgets}) and the \emph{positive} augmentation (depicted with blue dashed edges) does not match the exterior vertices $v_t$ and $v_b$ (Figure~\ref{fig:literal_gadget}).

	A \emph{literal gadget} is obtained by the union of two basic gadgets joined as shown in Figure~\ref{fig:literal_gadget}. As before, one can easily verify that there are only two possible plane topological augmentations (plane matchings) to change the parities of all red interior vertices. Note that if we choose the (negative) positive augmentation for one of the two basic gadgets, then we must also choose the (negative) positive augmentation for the other one. Figure~\ref{fig:literal_gadget} illustrates the \emph{positive} augmentation of a literal gadget, in which positive augmentations for the basic gadgets are used in the augmentation, and where the exterior vertices $v_t$ and $v'_t$ remain without changing their parities. In the \emph{negative} augmentation of a literal gadget, negative augmentations for the two basic gadgets appear and $v_l$ and $v'_r$ remain without changing their parities. The two exterior red vertices $v_l, v'_r$ will be used to concatenate literal gadgets corresponding to the same variable, one after another. The two exterior red vertices $v_t, v'_t$, called the \emph{output} of the literal gadget,  will be used to connect the literal gadget to a clause gadget (not defined yet).

	Observe that in the literal gadget shown in Figure~\ref{fig:literal_gadget}, the two bottom red vertices, each in a different basic gadget, are enclosed in a quadrilateral face, in such a way that if both vertices are free, then they are forced to join each other in order to meet their parity constraints. This kind of structure that consists of eight vertices (only two of them being red) and encloses its two red vertices is called a \emph{wire gadget}. In Figure~\ref{fig:literal_gadget}, the wire gadget consists of the three bottom vertices of each basic gadget, plus two extra blue vertices in the middle of the gadget to define the quadrangular face. Wire gadgets will be also used to connect literal gadgets to clause gadgets.

	\begin{figure}[ht!]
\centering
		\begin{subfigure}[t]{0.25\textwidth}
			\centering
			\includegraphics[width=.80\linewidth]{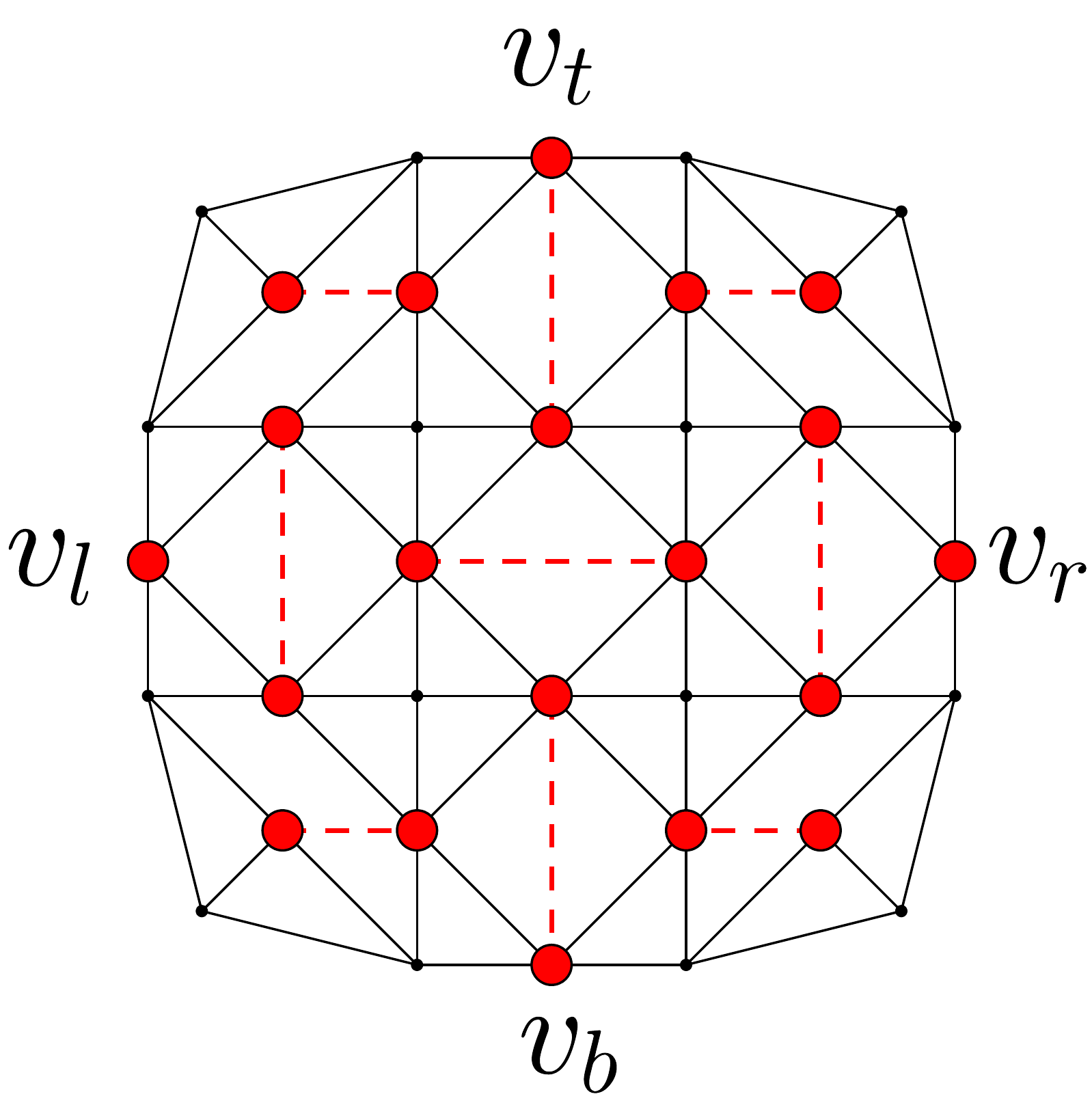}
			\caption{}
			\label{fig:basic_gadgets}
		\end{subfigure}%
		~
		\begin{subfigure}[t]{0.4\textwidth}
			\centering
			\includegraphics[width=.84\linewidth]{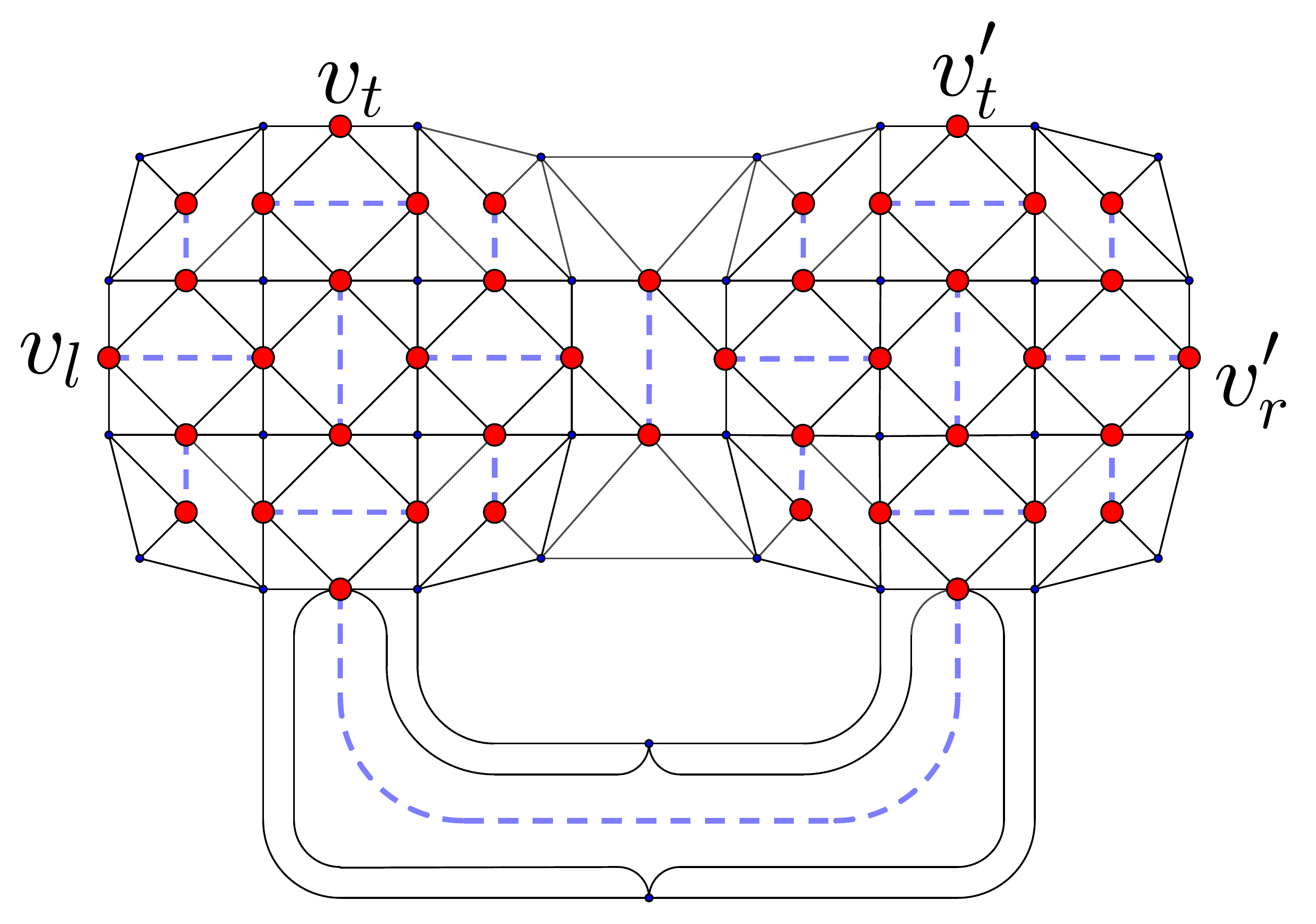}	
			\caption{}
			\label{fig:literal_gadget}
		\end{subfigure}%
		~
		\begin{subfigure}[t]{0.3\textwidth}
			\centering
			\includegraphics[width=.75\linewidth]{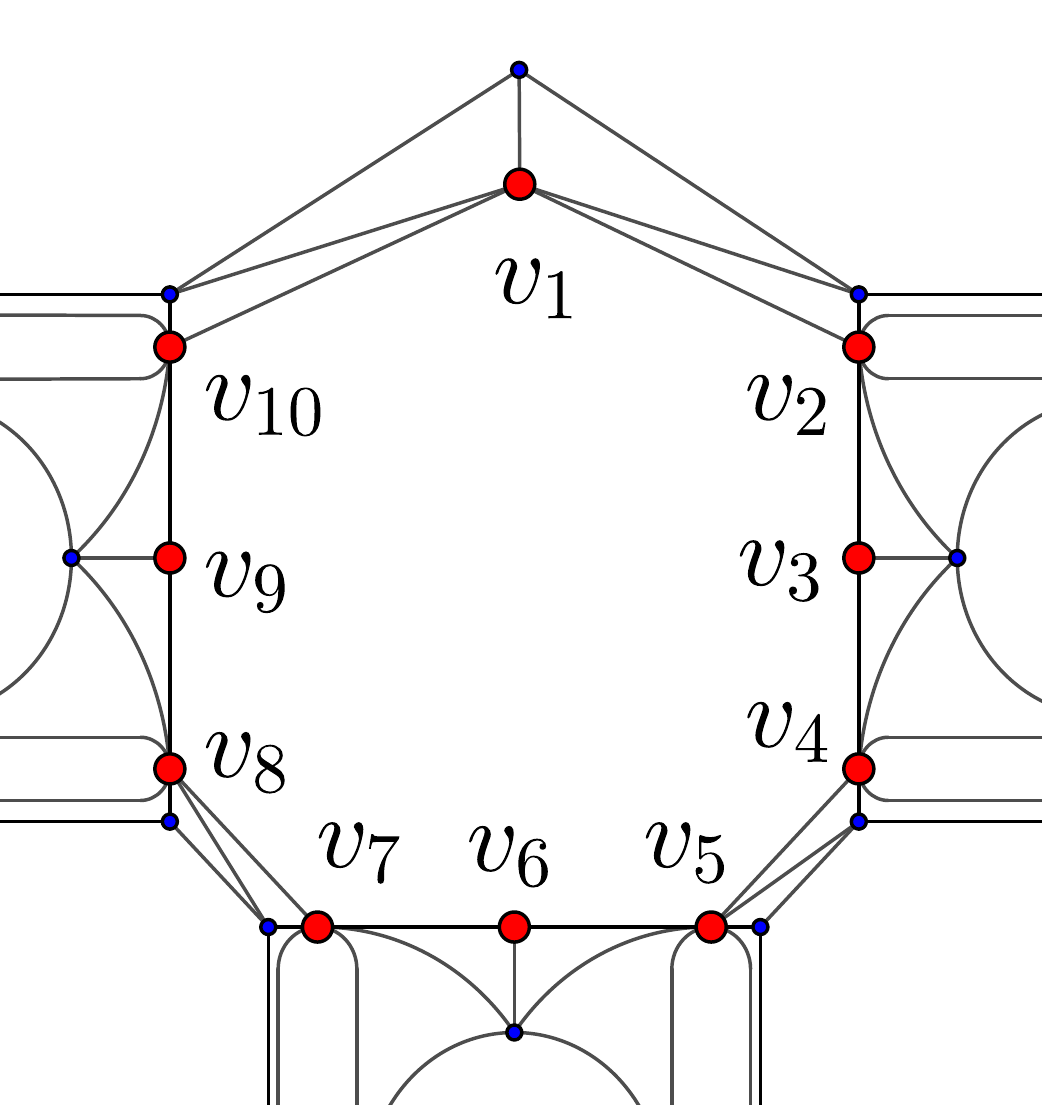}	
			\caption{}
			\label{fig:clause_gadget}
		\end{subfigure}%
		\caption{(a) The basic gadget. It admits only two possible plane augmentations of its red interior vertices. The negative augmentation is shown with red dashed edges. (b) A literal gadget with a positive augmentation. (c) A clause gadget.}
		\label{fig:gadgets}
	\end{figure}

		\begin{figure}[ht!]
		\centering
		\begin{subfigure}[t]{0.45\textwidth}

			\includegraphics[width=.78\linewidth]{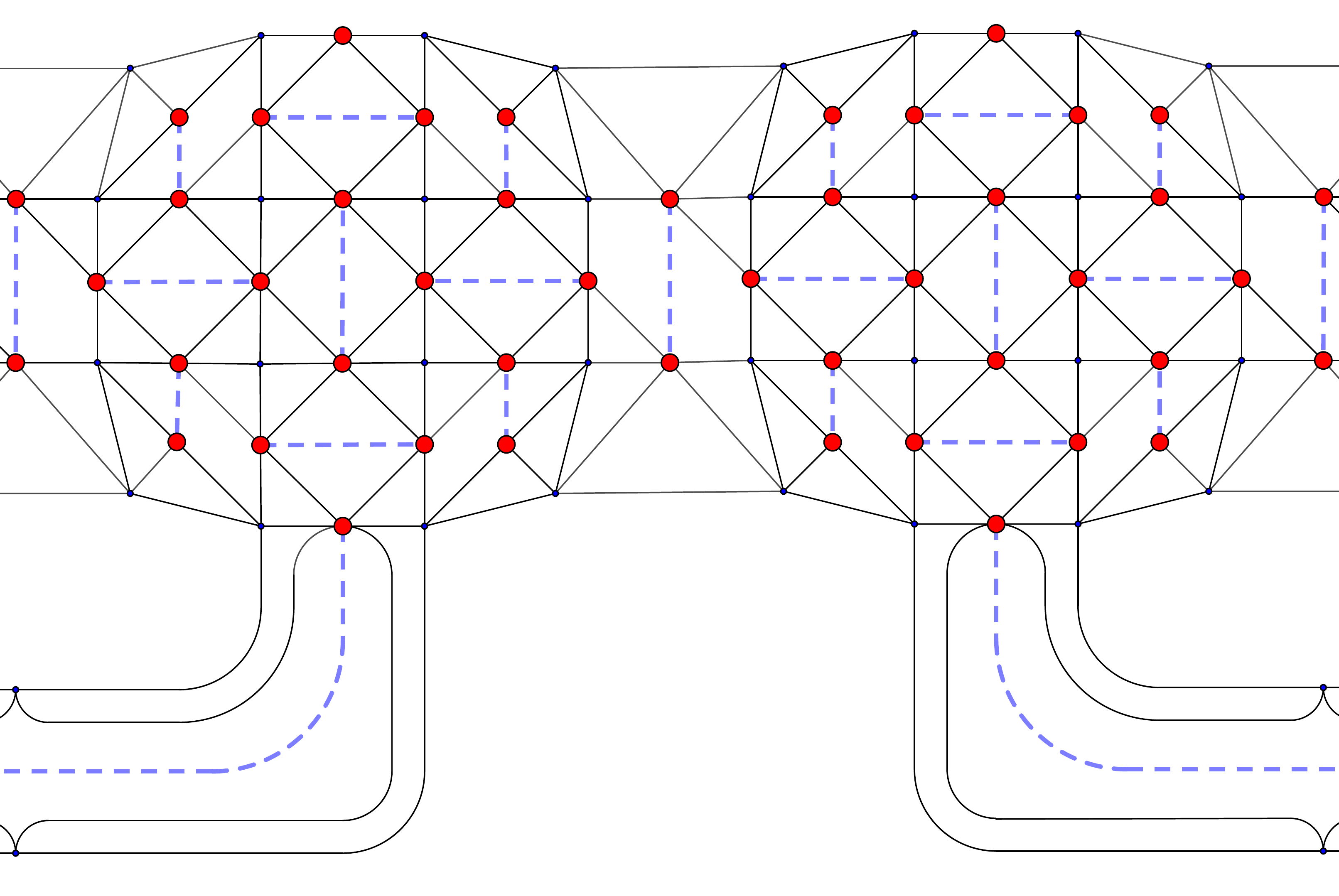}	
			\caption{}
			\label{fig:literal_gadget1}
		\end{subfigure}%
		~
		\begin{subfigure}[t]{0.45\textwidth}
			\centering
			\includegraphics[width=.78\linewidth]{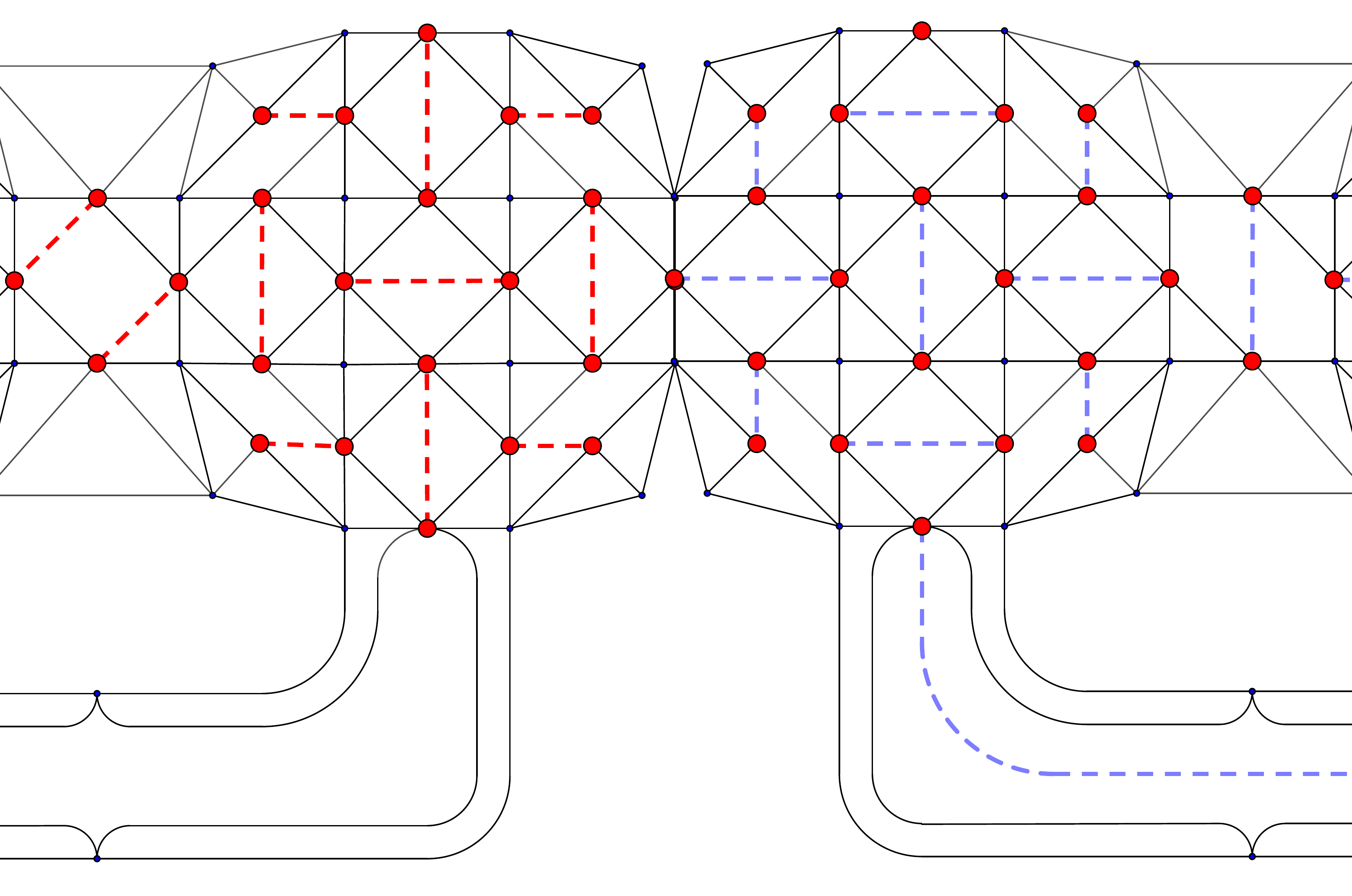}	
			\caption{}
			\label{fig:literal_gadget2}
		\end{subfigure}%
	
		\caption{(a) Union of two literal gadgets when their corresponding consecutive literals in $L_x$ are both positive or negative. (b) Union of two literal gadgets when one of their corresponding literals is positive and the other negative.}
	\end{figure}

Let $L_x=\{l_1, l_2, \ldots , l_k\}$ be the ordered list of literals of a variable $x$, according to $F$. We define the \emph{variable gadget} associated with $x$ as a sequence of $k$ literal gadgets joined as follows: If two consecutive literals $l_i$ and $l_{i+1}$ are both positive or negative, then their corresponding literal gadgets are joined as illustrated in Figure~\ref{fig:literal_gadget1}. Otherwise, they are joined as illustrated in Figure~\ref{fig:literal_gadget2}. In this way, the same augmentation for a literal gadget is transmitted to the following one in the first case, and the opposite augmentation is transmitted to the following literal gadget in the second case. We also join the leftmost exterior vertex of the first literal gadget to the rightmost exterior vertex of the last literal gadget with a wire gadget, when $l_1$ and $l_k$ are both positive or negative. If $l_1$ is positive and $l_k$ is negative or viceversa, then we join these two red vertices with a \emph{double wire gadget}, that is, a gadget obtained by gluing two wire gadgets
(see the variable gadget corresponding to $x_2$ in Figure~\ref{fig:3satphi}).

By construction, it is straightforward to see again that there are only two possible plane topological augmentations to change the parities of all red interior vertices in a variable gadget. In addition, all literal gadgets corresponding to positive literals in $L_x$ must have the same literal gadget augmentation, while all literal gadgets corresponding to negative literals must have the opposite one. See Figure~\ref{fig:3satphi}. The augmentation of a variable gadget with positive augmentations for the literal gadgets corresponding to positive literals in $L_x$ is called \emph{positive}, and \emph{negative} otherwise.

	A \emph{clause gadget} is a graph as shown in Figure \ref{fig:clause_gadget}, whose main part is a decagon having all its vertices, say from $v_1$ to $v_{10}$, colored in red. We identify the vertices $v_2$ and $v_4$ as the first input, $v_5$ and $v_7$ as the second input, and $v_8$ and $v_{10}$ as the third input. An input will be connected to the output of a literal gadget as we will explain later.

From $F$, we build $G_{\Phi}$ as follows. A clause in $F$ is replaced by a clause gadget. A variable $x$ in $F$ is replaced by a variable gadget, where the order of the literal gadgets corresponds to the order of the literals in $L_x$. An edge in $F$ is replaced by two wire gadgets connecting the red vertices of an input to the red vertices of an output, Finally, the exterior of every gadget is triangulated. See Figure~\ref{fig:3satphi}. $C_{G_{\Phi}}$  is chosen such that $R(C_{G_{\Phi}})$ is the set of red vertices of $G_{\Phi}$. Note that $G_{\Phi}$ is topological and plane and any edge that is added to $G_{\Phi}$ must be inside a gadget.

	\begin{figure}[ht!]
		\centering
		\includegraphics[width=\textwidth] {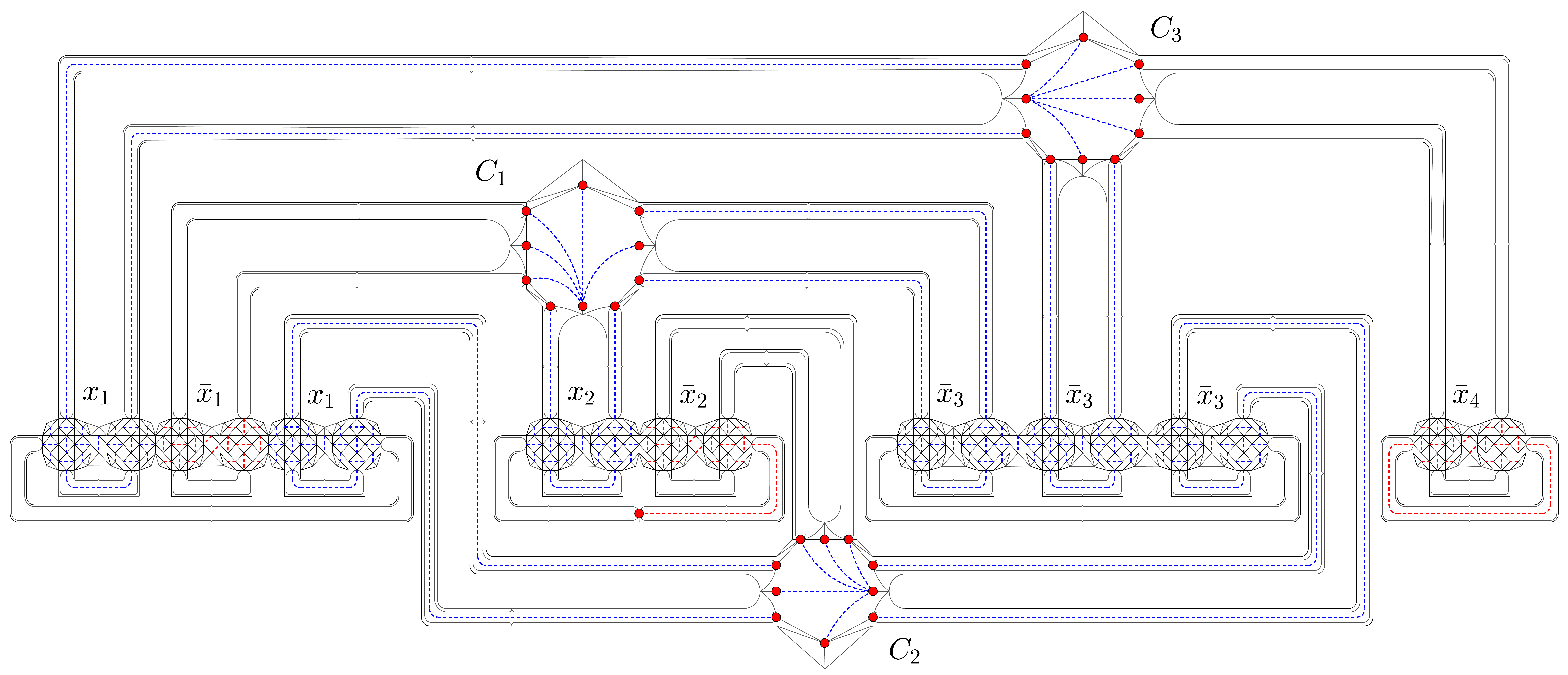}
		\caption{The graph $G_{\Phi}$ for the formula $\Phi$ =
			$(\bar{x}_1 \vee x_2 \vee \bar{x}_3)
			\wedge(x_1 \vee \bar{x}_2 \vee \bar{x}_3)
			\wedge (x_1 \vee \bar{x}_3 \vee \bar{x}_4)$ and the plane embedding $F$ of $F_{\Phi}$ shown in Figure~\ref{fig:3sat}. The exterior of the gadgets is triangulated. The plane topological augmentation shown in the figure correspond to the following assignment $(x_1, x_2, x_3, x_4) = (T,T,F,T)$.}
		\label{fig:3satphi}
	\end{figure}%

From this construction, one can easily prove that $G_{\Phi}$ is topologically augmentable to meet $C_G$ if and only if $\Phi$ is satisfiable. Suppose first that $\Phi$ is satisfiable. We can augment $G_{\Phi}$ as follows. If a variable $x$ takes the ``true" value, then we add to $G_{\Phi}$ a positive augmentation for the corresponding variable gadget. Otherwise, we add a negative augmentation. In a positive augmentation of a literal gadget, the two vertices of its output do not meet their parity constraints and must be connected to the two vertices of an input through the wire gadgets. Since $\Phi$ is satisfiable, at least one literal in each clause is true, implying that the two red vertices of the output of this true literal are connected to the two vertices, say $v_2, v_4$, of one of the inputs of the clause. Thus, we can add a star centered at $v_3$ to meet the parities of the red vertices in the clause gadget not yet meeting their parities. Hence, $G_{\Phi}$ is augmentable.

On the contrary, suppose that $G_{\Phi}$ is topologically augmentable to meet $C_G$. A variable gadget corresponding to a variable $x$ can be augmented only by a positive or negative augmentation. This implies the ``true" value for $x$ in the first case and the ``false" value in the second case. On the other hand, the parities of the ten red vertices of a clause gadget cannot be changed  using only edges inside the decagon. The decagon and the addition of any set of non-crossing diagonals inside it define a biconnected outerplanar graph. It is well known that a biconnected outerplanar graph always has at least two vertices of degree 2. Thus, if there is a plane augmentation for a clause gadget, then at least two vertices must be connected to output vertices through the wire gadgets to meet their parities. Given a gadget clause corresponding to a clause $c$, if a vertex of an input, say $v_2$, is connected in $G_{\Phi}$ to a vertex of an output, necessarily $v_4$ is also connected to the same output. Hence, the literal gadget containing this output will have a positive augmentation in $G_{\Phi}$, implying that $c$ has at least a true literal. Therefore, $\Phi$ is satisfiable and the theorem follows.
\end{proof}

In the previous theorem, $R(C_G) \subset V$. Let us strengthen the previous result to the special case $R(C_G) = V$.

\begin{theorem}
	Let $G=(V,E)$ be a plane topological graph and let $C_G$ be a set of parity constraints. Suppose that $R(C_G) = V$. Then, the problem of deciding if $G$ is topologically augmentable to meet $C_G$
is $\mathcal{NP}$-complete.
	\label{teo:np_topo_planes_v}
\end{theorem}
\begin{proof}
The proof is again by reduction from Planar 3-SAT. Given a 3-SAT formula $\Phi$ and a plane embedding $F$ of $F_{\Phi}$, we build in polynomial time a plane topological graph $G'_{\Phi}$ such that $R(G'_{\Phi}) = V$ and $G'_{\Phi}$ is topologically augmentable to meet $C_{G'_{\Phi}}$ if and only if $\Phi$ is satisfiable. Equivalently, we prove that $G'_{\Phi}$ is topological augmentable if and only if $G_{\Phi}$ is topologically augmentable, where $G_{\Phi}$ is the graph built in Theorem~\ref{teo:np_topo_planes}.

From $G_{\Phi}$, we show how to change all blue vertices in $G_{\Phi}$ to red vertices in $G'_{\Phi}$, keeping the rest of the vertices as red vertices. Each blue vertex in $G_{\Phi}$ is adjacent to at least one triangular face, so we can do a mapping of each blue vertex in $G_{\Phi}$ to one of its adjacent triangular faces in $G_{\Phi}$.
	
	Let $\Delta = (v_1, v_2, v_3)$ be a triangular face in $G_{\Phi}$. We have the following two cases: Only one blue vertex, say $v_1$, is assigned to $\Delta$, or more than one blue vertex is assigned to $\Delta$. In the first case, we can change the color of $v_1$ and preserve the parities of $v_2,v_3$, by adding three red vertices $v'_1, v'_2, v'_3$ and connecting them as shown in Figure~\ref{fig:red}. Observe that in any plane augmentation, $v'_2$ and $v'_3$ are forced to join each other to meet their parities. In the same way, $v_1$ is forced to join $v'_1$, assuming that $v_1$ is red. In the second case, more than one blue vertex is assigned to $\Delta$, suppose that $v_1$ and $v_2$ are blue. We can change the color of $v_1$ by adding the previous construction in $\Delta$, and the color of $v_2$ by adding the construction in the triangle $(v_3,v'_3,v_2)$. If $v_3$ is also blue, then we add again the construction in the triangle $(v'_1,v'_3,v_3)$ to change the color of $v_3$.

\begin{figure}[ht!]
        \centering
		\begin{subfigure}[t]{0.3\textwidth}
			\centering
			\includegraphics[width=.8\linewidth]{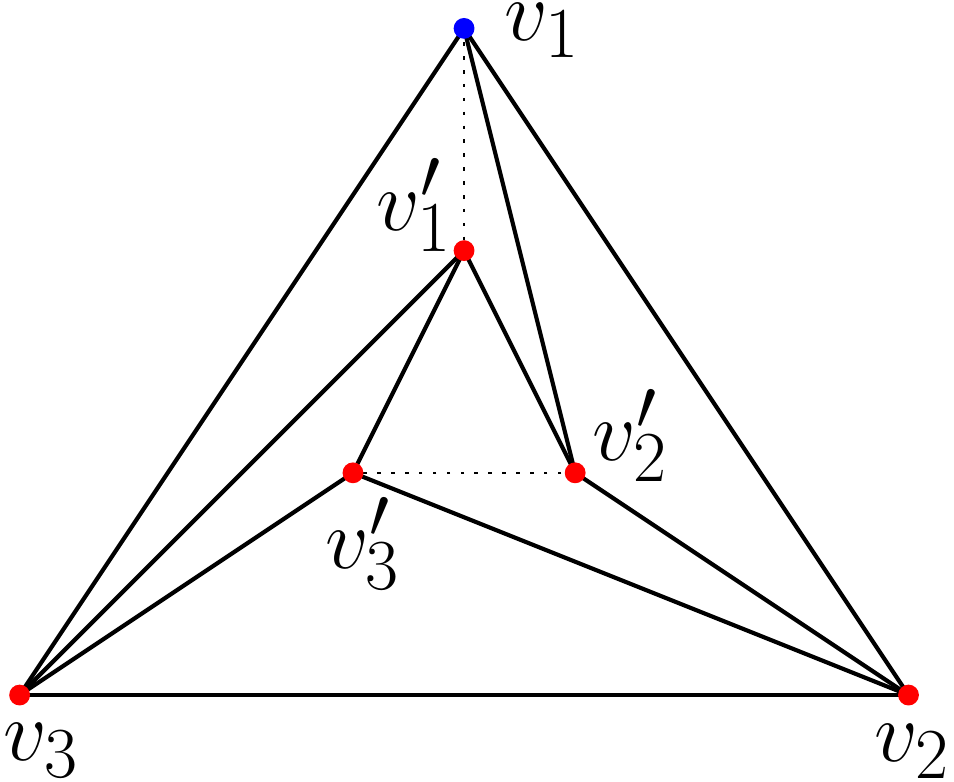}
			\caption{}
			\label{fig:red}
		\end{subfigure}%
		~~~~~~~~~~~~~~
		\begin{subfigure}[t]{0.3\textwidth}
			\centering
	   \includegraphics[width=.9\linewidth]{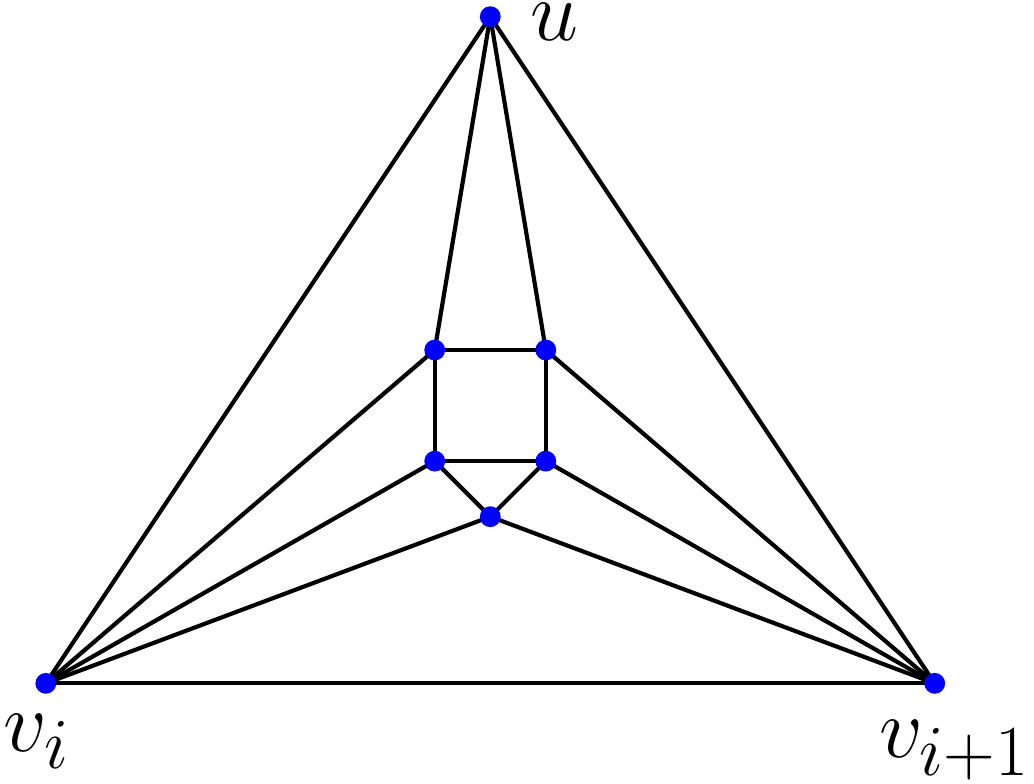}
			\caption{}
			\label{fig:Figura2}
		\end{subfigure}%
		\caption{(a) Changing the color of $v_1$ by adding three red vertices in a triangular face $\Delta$. (b) The triangle $(v_i,v_{i+1},u)$ with the addition of 5 blue vertices and 14 edges, to replace the duplicated edge $e'_i=(v_i,v_{i+1})$.}
		\label{fig:bluered}
	\end{figure}

	 Obviously, the graph $G'_{\Phi}$ obtained after applying the previous rules to all  triangular faces of $G_{\Phi}$, is plane and satisfying $R(G'_{\Phi}) = V$. In addition, it is straightforward to see that $G'_{\Phi}$ is topologically augmentable if and only if $G_{\Phi}$ is topologically augmentable, since all red vertices in $G'_{\Phi}$ that do not belong to $G_{\Phi}$ must be matched between them to change their parities in any augmentation. Thus, the theorem follows.
\end{proof}

A similar reasoning applies to prove that the problem of deciding if a plane topological graph can be augmented to a Eulerian plane topological graph is $\mathcal{NP}$-complete. Before proving this result, we recall the concept of $T$-join. Let $G=(V,E)$ be a graph and let $T\subseteq V$. A subset $E'\subseteq E$ of edges is a \emph{$T$-join} if the set of vertices with odd degree in the subgraph induced by $E'$ is precisely $T$. If $G$ is connected and $|T|$ is even, then $G$ always has a $T$-join. In this case, a $T$-join of minimum size consists of $|T|/2$ edge-disjoint paths whose endpoints are the vertices in $T$, and can be computed in $\mathcal{O}(n^3)$~\cite{edmonds73}.

\begin{theorem}
Let $G$ be a plane topological graph. Then, the problem of deciding if $G$ can be augmented to a Eulerian plane topological simple graph is $\mathcal{NP}$-complete.
\label{teo:np_topo_euler}
\end{theorem}
\begin{proof}
Given a 3-SAT formula $\Phi$ and a plane embedding $F$ of $F_{\Phi}$, let $G_{\Phi}= (V,E)$ be the graph built in Theorem~\ref{teo:np_topo_planes}. We show that we can transform $G_\Phi$ into a plane topological graph $G'_\Phi$, having all its red vertices with odd degree and all its blue vertices with even degree, in such a way that $G'_{\Phi }$ is topologically augmentable if and only if $G_{\Phi }$ is topologically augmentable. We recall that to obtain a Eulerian graph, we have to change the parities of the odd degree vertices.

Let $V'$ be the set of vertices with odd degree in $G_{\Phi}$. Let $T$ be the set of vertices $v_i \in V$, such that $v_i$ is red and has even degree or $v_i$ is blue and has odd degree. The number of vertices in $T$ is even because $|R(G_{\Phi})|$ and $|V'|$ are even. $G_{\Phi}$ is connected, so we can find a $T$-join of minimum size $G_T$ in $\mathcal{O}(n^3)$ time, which consists of $|T|/2$ edge-disjoint paths whose endpoints are the vertices in $T$.

We claim that any triangular face $\Delta =(v_i,v_j,v_k)$ of $G_{\Phi}$ has at most one common edge with $G_T$. Suppose to the contrary that $\Delta$ has at least two common edges, say $(v_i,v_j), (v_j,v_k)$. If  they belong to the same path $P_1, (v_i,v_j), (v_j,v_k), P_2$ of $G_T$, then replacing the two common edges by the third edge we obtain a shorter path $P_1, (v_i,v_k), P_2$, contradicting the minimality of $G_T$. If they belong to different paths, $P_1, (v_i,v_j), P_2$ and $P'_1, (v_j,v_k), P'_2$, then $P'_1, P_2$ and $P_1,(v_i,v_k), P'_2$ also connect the same endpoints and are shorter, again a contradiction. Hence, the claim follows.

Let $G'$ be the multigraph obtained by adding $G_T$ to $G_{\Phi}$, so $G'$ contains duplicated edges. Note that each edge of $G_{\Phi}$ is adjacent to at least a triangular face and that the vertices with odd degree in $G'$ are precisely the red vertices. By the previous claim, each duplicated edge in $G'$ can be mapped to a different triangular face and embedded into it. Suppose that $e'_i=(v_i,v_{i+1})$ is a duplicated edge embedded inside the triangular face $(v_i,v_{i+1},u)$. Then, we can replace $e'_i$ by 5 blue vertices and 14 edges and connect them as shown in Figure \ref{fig:Figura2}. Observe that the 5 added vertices have even degree, and $v_i,v_{i+1}, u$ keep their parities. By doing this replacement for every duplicated edge of $G'$, we obtain a new plane topological simple graph $G'_{\Phi }$ such that its red vertices are precisely the ones with odd degree. By construction, it is straightforward to verify that $G'_{\Phi }$ is topologically augmentable if and only if $G_{\Phi }$ is topologically augmentable, so the theorem follows.
\end{proof}

		\begin{figure}[ht!]
\centering
		\begin{subfigure}[t]{.35\textwidth}
			\centering
			\includegraphics[width=.95\textwidth] {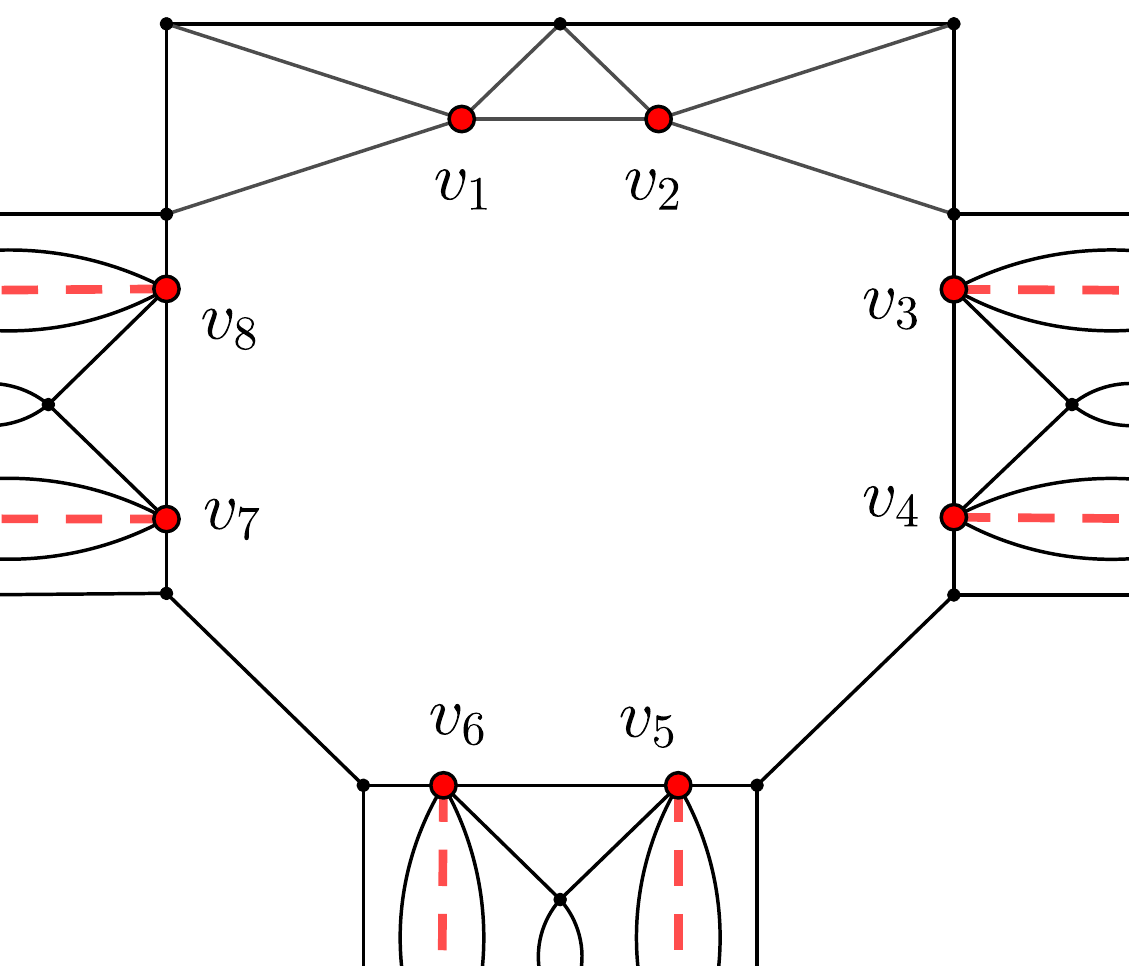}
			\caption{}
			\label{fig:clause_gadget_t}
		\end{subfigure}%
		~~~~~~~~~~~
		\begin{subfigure}[t]{.35\textwidth}
			\centering
			\includegraphics[width=.95\textwidth] {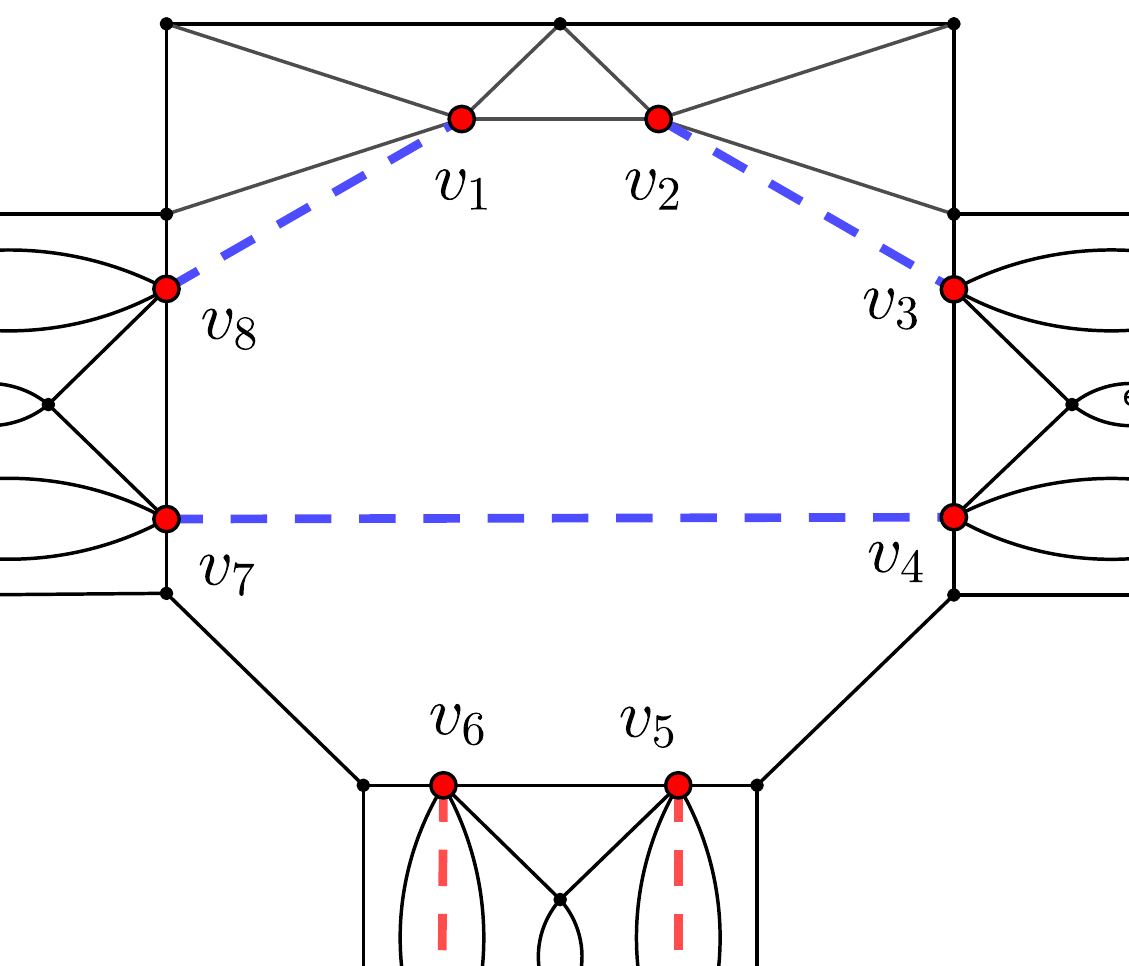}
			\caption{}
			\label{fig:clause_gadget_t2}
		\end{subfigure}%
		\caption{(a) A clause gadget receiving three negative values. (b) A clause gadget receiving at least one positive value.}
		\label{fig:gadgets_tree}
	\end{figure}

Next, we present analogous results to the previous ones, when the added graph is required to be a plane topological perfect matching $M$ between the red vertices.
Such a matching, if it exists, is the optimal way of augmenting a graph, since any graph with $2k$ red vertices cannot be augmented with less than $k$ edges.

\begin{theorem}
	Let $G=(V,E)$ be a plane topological graph and let $C_G$ be a set of parity constraints. Then, the problem of deciding if $G$ is topologically augmentable to meet $C_G$, by a plane topological perfect matching $M$ between the vertices in $R(C_G)$,
is $\mathcal{NP}$-complete.\footnote{This result was presented at the 29th Canadian Conference on Computational Geometry, Ottawa (2017)~\cite{aldana2017planarity}.}
	\label{teo:topomin}
\end{theorem}
\begin{proof}
	The proof is based again on a reduction from Planar 3-SAT. Given a 3-SAT formula $\Phi$ and a plane embedding $F$ of $F_{\Phi}$, we build in polynomial time a plane topological graph $G_{\Phi}$ and a set $C_{G_{\Phi}}$ of parity constraints, such that $G_{\Phi}$ is topologically augmentable to meet $C_{G_{\Phi}}$, by a plane topological perfect matching between the red vertices of $G_{\Phi}$, if and only if $\Phi$ is satisfiable.

	We construct $G_\Phi$ in a similar way as described in Theorem~\ref{teo:np_topo_planes}, replacing some of its gadgets as follows, so that $G_\Phi$ is topological and plane. Variable gadgets are the same as described in the proof of Theorem~\ref{teo:np_topo_planes}. A clause gadget is now a subgraph as shown in Figure \ref{fig:clause_gadget_t}, whose main part consists of a 14-gon having eight red vertices. We identify the vertices $v_3, v_4$ as the first input, $v_5, v_6$ as the second output and $v_7, v_{8}$ as the third input. The vertices in a input can change their parities by only two options, either by adding an edge to a red vertex in the wire gadget or being adjacent to a red vertex inside the face of the clause. The remaining two red vertices in the clause, $v_1$ and $v_2$, can change their parities only by the addition of edges lying on the face of the clause. Wire gadgets connecting literals gadgets and clause gadgets are now replaced by double wire gadgets. Thus, if a positive augmentation is used for a literal gadget, then the only way of changing the parities of the vertices in the output of this literal gadget is connecting them to the middle red vertices in the double wire gadgets. This implies that the vertices in the corresponding input must change their parities by adding edges in the face of the clause to which they belong. On the contrary, if a negative augmentation is used for a literal gadget, then the red vertices in the middle of the double wire gadgets must change their parities by connecting them to the vertices of the input. Finally, the exterior of the gadgets is triangulated.

	Observe that if a clause gadget receives three negative values (negative augmentations in its corresponding literal gadgets), then there is no way of changing the parities of the vertices $v_1$ and $v_2$ with less than two edges, as the clause gadget shown in Figure \ref{fig:clause_gadget_t} depicts. Otherwise, if at least one input receives a positive value (a positive augmentation in its corresponding literal gadget), then the remaining red vertices in the clause can meet their parities by doing a plane matching between them. An example is shown in Figure \ref{fig:clause_gadget_t2}.
	
The new graph $G_{\Phi}$ obtained after this replacement is clearly plane. In addition, by the previous discussion, one can easily verify that $G_{\Phi}$ is topologically augmentable by a plane topological perfect matching between the red vertices if and only if $\Phi$ is satisfiable, hence the theorem follows.
\end{proof}

\begin{theorem}
	Let $G=(V,E)$ be a plane topological graph. Then, the problem of deciding if there exists a plane topological perfect matching $M$
such that $G$ and $M$ are compatible and edge-disjoint is $\mathcal{NP}$-complete.
	\label{teo:topominall}
\end{theorem}
\begin{proof}
In this case, the set of red vertices of $G=(V,E)$ is precisely $V$, and we look for a plane topological perfect matching $M$ that is compatible and disjoint. From the graph $G_{\Phi}$ built in the previous theorem, observe that we only need to show how to change the color of every blue vertex in $G_{\Phi}$ to build a new graph $G'_{\Phi}$, such that $G_{\Phi}$ is topologically augmentable by a plane topological perfect matching between the its red vertices, if and only if, $G'_{\Phi}$ is topologically augmentable by a plane topological perfect matching between its vertices. As in the proof of Theorem~\ref{teo:np_topo_planes_v}, this can be done  using the gadget shown in Figure~\ref{fig:red} to change the color of every blue vertex in $G_{\Phi}$.
\end{proof}

\begin{theorem}
	Let $G=(V,E)$ be a plane topological graph. Then, the problem of deciding if $G$ can be augmented to a Eulerian plane topological simple graph, by the addition of a plane topological perfect matching $M$ between the vertices with odd degree in $G$, is $\mathcal{NP}$-complete.
	\label{teo:topominodd}
\end{theorem}
\begin{proof}
In this case, the set of vertices with odd degree is the set of red vertices, and we look for a plane topological perfect matching $M$ between the vertices with odd degree. To prove the $\mathcal{NP}$-completeness of this problem, we proceed as in the proof of Theorem~\ref{teo:np_topo_euler}. From the graph $G_{\Phi}$ built in the proof of Theorem~\ref{teo:topomin}, we add a minimum $T$-join to $G_{\Phi}$, where $T$ is the set of red vertices with even degree and the set of blue vertices with odd degree, and we replace every duplicated edge by the gadget shown in Figure~\ref{fig:Figura2}. In the resulting graph $G'_{\Phi}$, every vertex with odd degree is red and every vertex with even degree is blue. It is straightforward to see that $G_{\Phi}$ admits a plane topological perfect matching between its red vertices, if and only, if $G'_{\Phi}$ admits a plane topological perfect matching between its vertices with odd degree, so the theorem holds.
\end{proof}

To finish this section, we point out that the previous results can be extended to planar graphs, using the Whitney's theorem on 3-connected planar graphs.

\begin{theorem}
Let $G=(V,E)$ be a planar graph and let $C_G$ be a set of parity constraints. Then, the problem of deciding if there exists a planar graph $H$ on the same vertex set $V$, such that $G$ and $H$ are edge-disjoint, $G'=G\cup H$ is planar and $G'$ meets $C_G$, is $\mathcal{NP}$-complete. The problem is also $\mathcal{NP}$-complete when $H$ must be a perfect matching between the vertices in $R(C_G)$. In addition, the problem remains $\mathcal{NP}$-complete in both variants, even when $R(C_G) = V$ or $R(C_G)$ is the set of vertices with odd degree in $G$.
\label{teo:np_topo_planes_2}
\end{theorem}

\begin{proof}
It can be easily checked that the graph $G_\Phi$ built in the proof of Theorem~\ref{teo:np_topo_planes} is 3-connected. By Whitney's theorem, a 3-connected planar graph has a unique embedding (up to the choice of the unbounded face). Hence, we can consider $G_\Phi$ as a planar graph that is augmentable if and only if its unique plane embedding is topologically augmentable. Therefore, the first part of the theorem holds. In addition, the different graphs $G_\Phi$ or $G'_{\Phi}$ built in the proofs of Theorems~\ref{teo:np_topo_planes_v},~\ref{teo:np_topo_euler},~\ref{teo:topomin},~\ref{teo:topominall} and~\ref{teo:topominodd} are also 3-connected. Thus, the rest of the theorem also holds.
\end{proof}

\section{Plane geometric augmentation problems}\label{sec:geometric}
	
In this section we address the plane geometric augmentation problem to meet a set of parity constraints, that is, given a plane geometric graph $G=(V,E)$ and a set $C_G$ of parity constraints, we look for a plane geometric graph $H$ such that $G$ and $H$ are compatible and edge-disjoint, and all parity constraints are satisfied in $G\cup H$. We extend the $\mathcal{NP}$-completeness results obtained for plane topological graphs to plane geometric graphs. Moreover, some of these variants remain $\mathcal{NP}$-complete even on more restrictive geometric graph families, like trees or paths.

The following theorem is a consequence of Tutte's theorem on convex faces of 3-connected planar graphs~\cite{Tutte1963}.

	\begin{figure}[ht!]
\centering
	\begin{subfigure}[t]{.4\textwidth}
		\centering
		\includegraphics[width=0.8\textwidth] {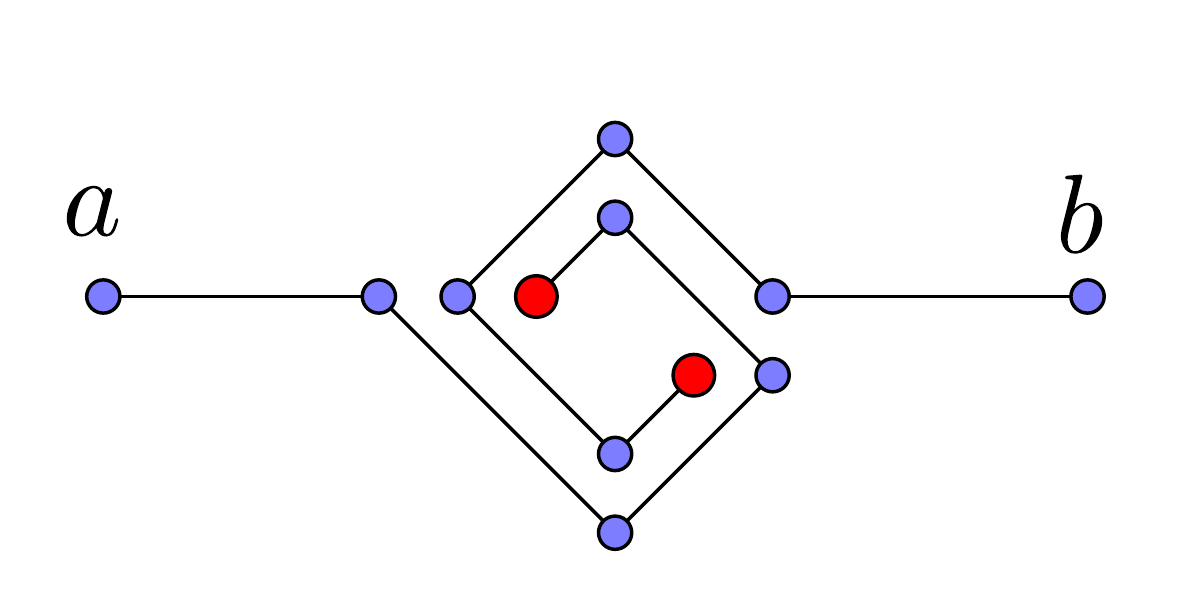}
		\caption{}
		\label{fig:gadget_spiral}
	\end{subfigure}%
	~~~~~~~~
	\begin{subfigure}[t]{.4\textwidth}
		\centering
		\includegraphics[width=.8\textwidth] {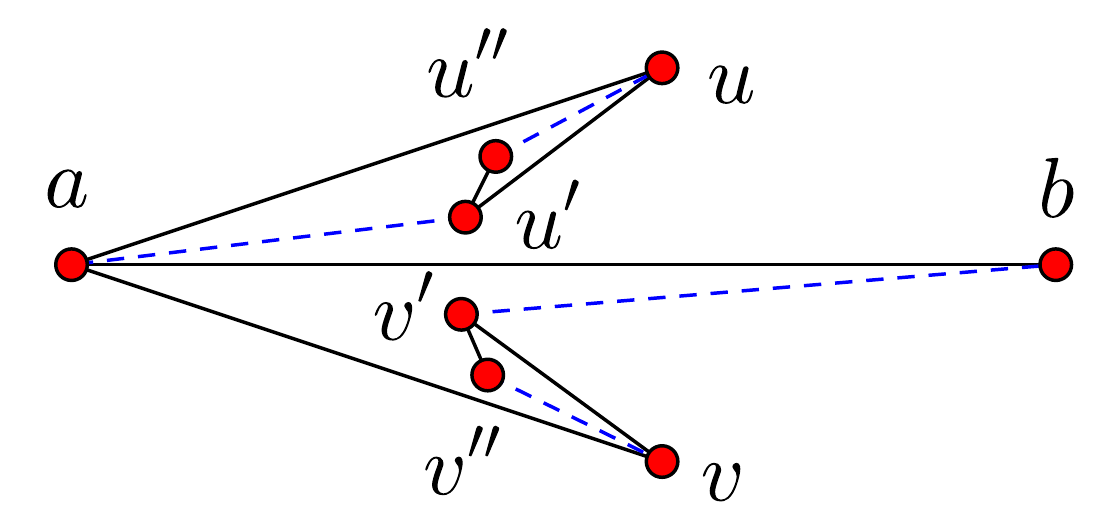}
		\caption{}
		\label{fig:gadget_arrow}
	\end{subfigure}%
	\caption{(a) Spiral gadget. (b) Arrow gadget.}
	\end{figure}

\begin{theorem}
	Let $G=(V,E)$ be a plane geometric graph and let $C_G$ be a set of parity constraints. Then, the problem of deciding if $G$ is geometrically augmentable to meet $C_G$
is $\mathcal{NP}$-complete. The problem is also $\mathcal{NP}$-complete when the augmenting graph must be a plane geometric perfect matching between the vertices in $R(C_G)$. In addition, the problem remains $\mathcal{NP}$-complete in both variants, even when $R(C_G) = V$ or $R(C_G)$ is the set of vertices with odd degree in $G$.
	\label{the:np_geom}
\end{theorem}
\begin{proof}
	By Tutte's theorem on convex faces of triconnected planar graphs, all graphs $G_{\Phi}$ and $G'_{\Phi}$ built in Theorems~\ref{teo:np_topo_planes}, \ref{teo:np_topo_planes_v},  \ref{teo:np_topo_euler}, \ref{teo:topomin}, \ref{teo:topominall} and \ref{teo:topominodd}, have an embedding where all edges are straight-line segments and all faces are convex. Then, they are geometrically augmentable to meet parity constraints if and only if $\Phi$ is satisfiable.
\end{proof}

In particular, from this theorem, deciding if a plane geometric graph can be geometrically augmented to Eulerian is $\mathcal{NP}$-complete, even when the added graph is required to be a plane geometric perfect matching between the vertices with odd degree.

Some previous variants remain $\mathcal{NP}$-complete for plane geometric trees and paths, as we show in the following theorems.

\begin{theorem}
	Let $T=(V,E)$ be a plane geometric tree and let $C_T$ be a set of parity constraints. Then, the problem of deciding if $T$ is geometrically augmentable by a plane geometric matching $M$ between the vertices in $R(G_T)$
is $\mathcal{NP}$-complete.
	\label{teo:np_geo_trees}
\end{theorem}
\begin{proof}

	We explain how to transform the graph $G_{\Phi}$ built in Theorem~\ref{teo:topomin} into a tree $T_{\Phi}$. As shown in \cite{rutter2008augmenting}, to break a cycle one can replace any edge $(a,b)$ of a cycle by the spiral gadget shown in Figure~\ref{fig:gadget_spiral}. Repeating this operation, we can break all cycles in $G_{\Phi}$, to obtain a plane geometric tree $T_{\Phi}$. Notice that the two red vertices in each spiral gadget are forced to join each other in any minimum augmentation. Therefore, $G_{\Phi}$ is geometrically augmentable with a plane geometric perfect matching between its red vertices, if and only if, $T_{\Phi}$ is geometrically augmentable with a plane geometric perfect matching between its red vertices, so the theorem follows.
\end{proof}

\begin{theorem}
	Let $T=(V,E)$ be a plane geometric tree. Then, the problem of deciding if there exists a plane geometric perfect matching $M$, such that $T$ and $M$ are compatible and edge-disjoint, is $\mathcal{NP}$-complete.
	\label{teo:np_geo_treesall}
\end{theorem}
\begin{proof}
Note that in this case, the set of red vertices is $V$. To prove that this problem is $\mathcal{NP}$-complete, we only need to show how to change the color of each blue vertex in the tree $T_{\Phi}$ built in the previous theorem, such that all new red vertices in the resulting plane geometric tree match between them in any minimum augmentation.
	
	Consider an edge $(a,b) \in T_\Phi$, with $a$ and $b$ being blue vertices. We can change the colors of $a$ and $b$ by using what we call an \emph{arrow gadget}. We add to $(a,b)$ two nested paths, each one consisting of three red vertices, placed at distance $\epsilon$ on each side of the edge as shown in Figure~\ref{fig:gadget_arrow}. The vertices of the path on the upper side of the edge, say $u,u',u''$, are placed in such a way that they cannot see the vertices of the path in the lower side, say $v,v',v''$. Moreover, the upper path is placed in such a way that $u'$ is visible only from $a,b$, and $u''$ is visible only from $u,a$. Symmetrically for the case of the lower path. 
Observe that $u''$ is forced to join $u$ and $v''$ is forced to join $v$, because otherwise the parities of the remaining vertices in the gadget are not met. Thus, the vertices in an arrow gadget must be matched between them in any perfect matching between the red vertices.

	Consider now an edge $(a,b) \in T_{\Phi}$ with $a$ being red and $b$ being blue. We can subdivide $(a,b)$ into two edges $(a,a')$ and $(a',b)$ and mark $a'$ as blue. Then, we can replace the edge $(a',b)$ by an arrow gadget as described previously. Iterating this replacement for all blue-blue edges and all blue-red edges in $T_{\Phi}$, we obtain a plane geometric tree $T'$ with all its vertices marked as red. Clearly, $T'$ has a plane geometric perfect matching, compatible and edge-disjoint with $T'$, if and only if, $T_{\Phi}$ can be geometrically augmented by a plane geometric perfect matching between its red vertices.
\end{proof}

\begin{figure}[ht!]
\centering
		\begin{subfigure}[t]{.25\textwidth}
			\centering
			\includegraphics[width=0.8\textwidth] {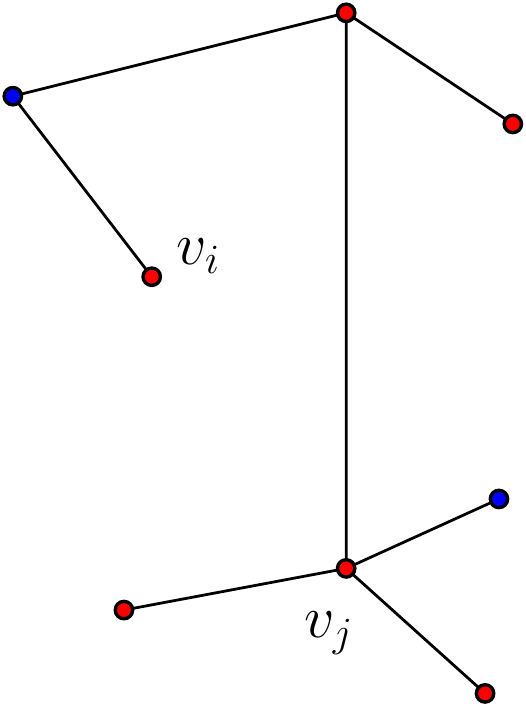}
			\caption{}
			\label{fig:tree_gadget}
		\end{subfigure}%
		~~
		\begin{subfigure}[t]{.25\textwidth}
			\centering
			\includegraphics[width=0.8\textwidth] {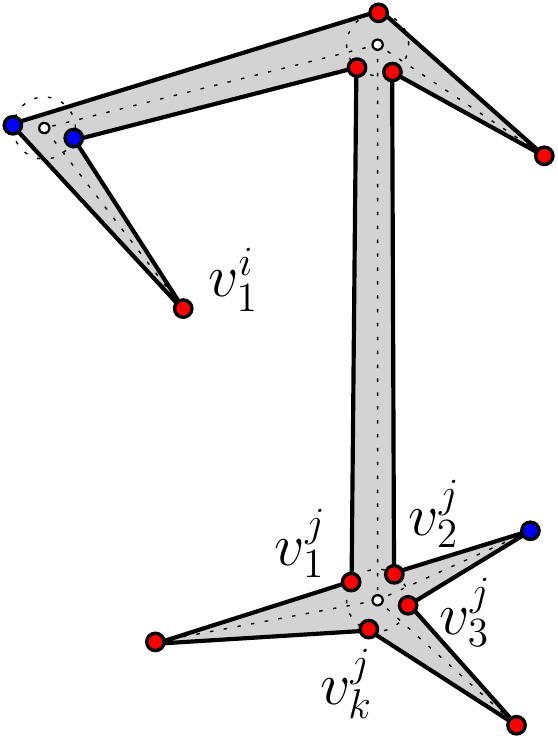}
			\caption{}
			\label{fig:path_gadget}
		\end{subfigure}%
		~~
		\begin{subfigure}[t]{.45\textwidth}
			\centering
			\includegraphics[width=.95\textwidth] {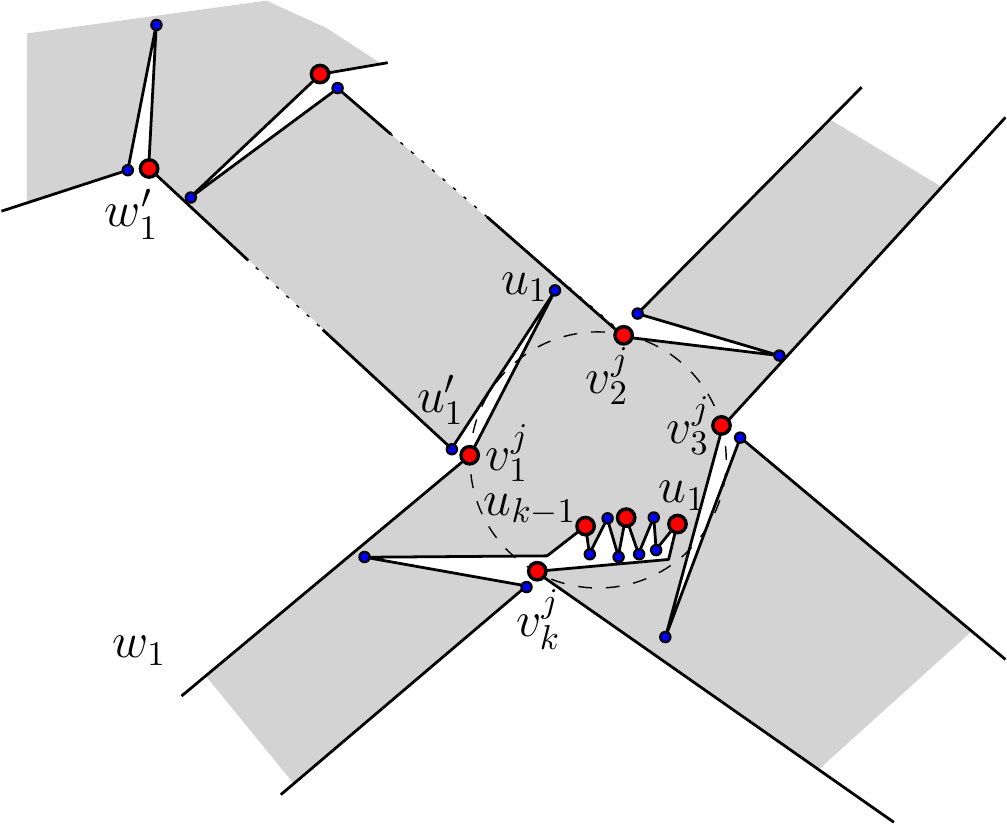}
			\caption{}
			\label{fig:path_gadget_augmented}
		\end{subfigure}%
		\caption{(a) A geometric tree $T_\Phi$. (b) A polygon $Q_\Phi$ obtained by traversing the unbounded face of $T_\Phi$. (c) Reconfiguration of $v_j$.}
	\end{figure}

In the following theorem, we establish the hardness of transforming a plane geometric tree  into Eulerian, by adding a plane geometric matching between the vertices with odd degree.

\begin{theorem}
	Let $T=(V,E)$ be a plane geometric tree. Then, the problem of deciding if $T$ can be augmented to a Eulerian plane geometric graph, by the addition of a plane geometric perfect matching $M$ between the vertices with odd degree in $T$, such that $T$ and $M$ are compatible and edge-disjoint, is $\mathcal{NP}$-complete.
	\label{teo:np_geo_trees_3}
\end{theorem}
\begin{proof}
Given a 3-SAT formula $\Phi$ and a plane embedding $F$ of $F_{\Phi}$, we can construct a plane topological graph as described in Theorem~\ref{teo:topomin}, that can be transformed into a plane geometric graph $G_{\Phi}$ with all faces being convex.
Then, we apply the same technique used in the proof of Theorem~\ref{teo:np_topo_euler}, and we add a $T$-join of minimum size to $G_{\Phi}$, where $T$ consists of all red vertices with even degree and all blue vertices with odd degree. After that, we replace every duplicated edge as indicated in Figure~\ref{fig:Figura2}, to obtain a new plane geometric graph $G'_{\Phi}$ such that all vertices with odd degree are red and all vertices with even degree are blue.

We now transform $G'_{\Phi}$ into a tree $T_{\Phi}$ by adding spiral gadgets to break cycles, as described in the proof of Theorem~\ref{teo:np_geo_trees}. Note that all vertices with odd degree in $T_{\Phi}$ are red and all vertices with even degree are blue. Clearly, by construction, $G_{\Phi}$ can be augmented by a plane geometric perfect matching between its red vertices, if and only if, $T_{\Phi}$ can be augmented by a plane geometric perfect matching between the vertices with odd degree. Therefore, the theorems follows.
\end{proof}

For plane geometric paths, we can prove the following results.

\begin{theorem}\label{teo:comp_match_trees}
	Let $P=(V,E)$ be a plane geometric path and let $C_P$ be a set of parity constraints. Then, the problem of deciding if $P$ is geometrically augmentable to meet $C_G$, by a plane geometric perfect matching $M$ between the vertices of $R(C_P)$,
is $\mathcal{NP}$-complete.
\end{theorem}
\begin{proof}
Given a 3-SAT formula $\Phi$ and a plane embedding $F$ of $F_{\Phi}$, we build a plane geometric tree $T_\Phi$ as described in Theorem \ref{teo:np_geo_trees}. We recall that $\Phi$ is satisfiable if and only if $T_{\Phi}$ can be augmented to meet parity constraints by a plane geometric perfect matching between the vertices in $R(T_{\Phi})$. From $T_{\Phi}$, we construct a plane geometric path $P_{\Phi}$ as follows.

It is well known that a simple polygon $Q_\Phi$ can be built from $T_{\Phi}$ by traversing the boundary of the unbounded face of $T_\Phi$, placing a copy infinitesimally close of a vertex every time it is visited, and connecting the copies according the traversal order, as shown in Figures \ref{fig:tree_gadget} and \ref{fig:path_gadget}. We start such a traversal at an arbitrary leaf $u_i$, and any copy of a vertex keeps the color of the vertex. For a non-leaf vertex, we can assume that the copies of the vertex are placed on a circumference of radius $\epsilon$ around the vertex.

	\begin{figure}[ht!]
		\begin{subfigure}[t]{\textwidth}
			\centering
			\includegraphics[width=.56\textwidth] {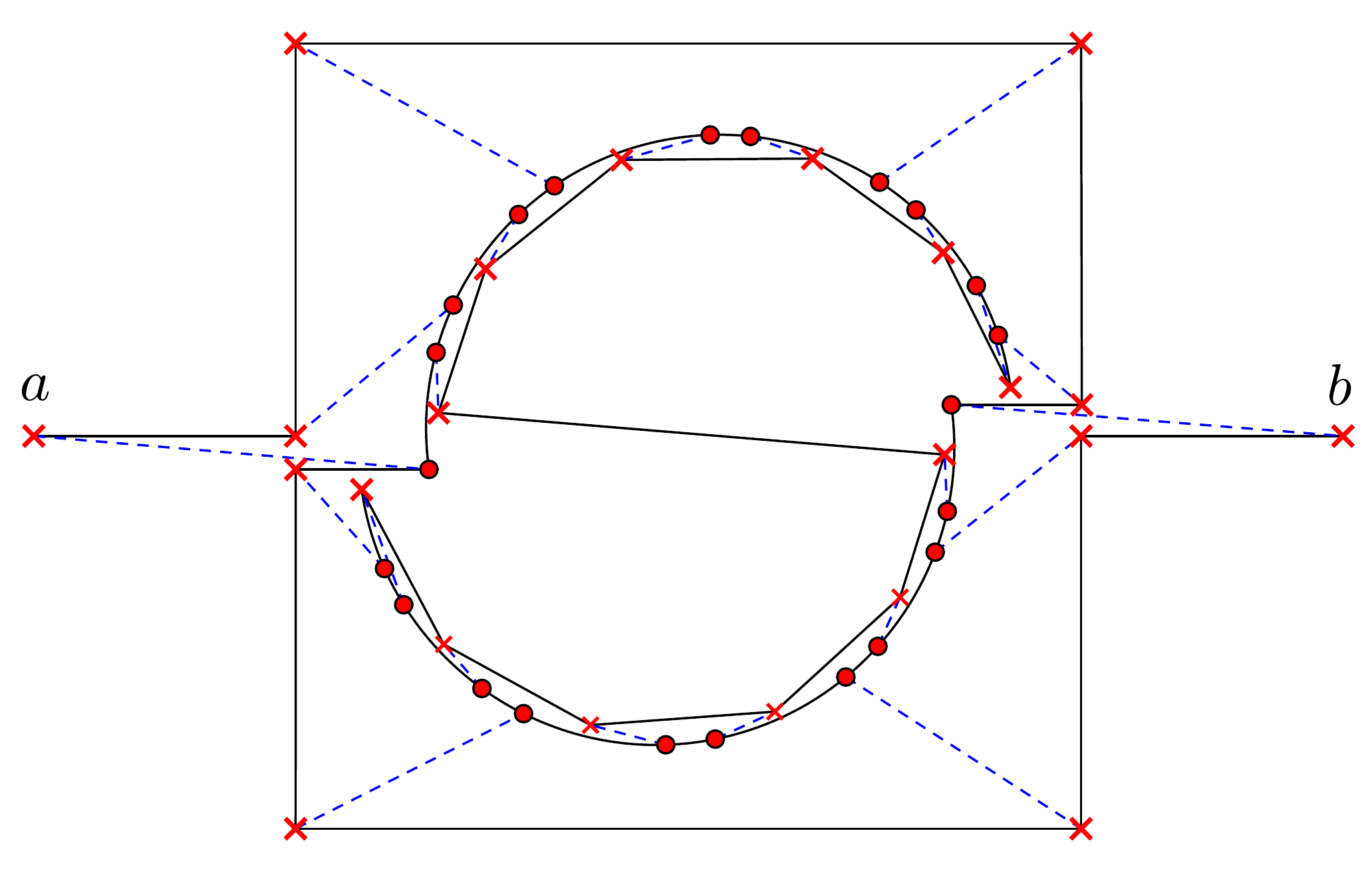}
			\caption{}
		\end{subfigure}
		\caption{A geometric path consisting of 22 disk-shaped vertices and 22 cross-shaped vertices.}
		\label{fig:path_gadget_red}
	\end{figure}

	For each red vertex $v_j \in T_\Phi$ with degree $k\geq 2$, we apply the following transformation to its corresponding vertices in $Q_\Phi$. An example of this transformation is shown in Figure~\ref{fig:path_gadget_augmented} for a vertex with degree 4 in $T_{\Phi}$. Let $C=\{v_1^j,v_2^j,...,v_k^j\}$ be the set of vertices in $Q_\Phi$ corresponding to a red vertex $v_j \in T_\Phi$, and assume that the clockwise order of these points along the boundary of $Q_\Phi$ is $\ldots , w_1, v_1^j, w'_1, \ldots , w_2, v_2^j, w'_2, \ldots , w_k, v_k^j, w'_k, \ldots$

For $l=1, \ldots ,k-1$, we add a narrow wedge $v_l^j, u_l, u'_l$ inside $Q_{\Phi}$ such that: (i) $u_l$ is placed very close to the edge $(w_{l+1},v_{l+1}^j)$, and (ii) $u'_l$ is placed very close to $v_l^j$ so that $v_l^j$ cannot be connected to $w'_l$, but $v_l^j$ can be connected to any other red vertex that was previously visible from it through the exterior of $Q_{\Phi}$. See Figure~\ref{fig:path_gadget_augmented}. At vertex $v_k^j$ we do the same operation, but also adding before the wedge a zig-zag path with most of the points in convex position and $k-1$ of them, $u_1, u_2, \ldots , u_{k-1}$, being red, in such a way that: (i) the only red vertices that are visible from a red vertex in the zig-zag belong to $C$, (ii) $u_1, u_2, \ldots , u_{k-1}$ are placed inside the circle of radius $\epsilon$ centered at $v_j$ so that every vertex $v_l^j \in C$, with $1\le l< k$, can see all of them, and (iii) $v_k^j$ is visible from at least one of $u_1, u_2, \ldots , u_{k-1}$.

Observe that the addition of these narrow wedges implies that a vertex $v_l^j\in C$ cannot be connected to any other red vertex $v_{l'}^{j'}\notin C$ through the interior of $Q_{\Phi}$, and that the red vertices in the zig-zag path must be matched to $k-1$ of the copies of $v_j$ in any augmentation. Thus, one of the copies of $v_j$ is always free to join any other red vertex in the rest of the graph through the exterior of $Q_{\Phi}$.

$Q_\Phi$ can be obviously transformed into a plane geometric path $P_\Phi$, by duplicating a blue vertex to break the boundary of $Q_\Phi$.	Then, it is straightforward to see that $P_\Phi$ admits a plane geometric perfect matching between its red vertices, if and only if, $T_{\Phi}$ can be augmented by a plane geometric perfect matching between its red vertices. Hence, the theorem holds.
\end{proof}

\begin{theorem}\label{teo:comp_match_paths}
	Let $P=(V,E)$ be a plane geometric path. Then, the problem of deciding if there exists a plane geometric perfect matching $M$ such that $P$ and $M$ are compatible and edge-disjoint is $\mathcal{NP}$-complete.
\end{theorem}
\begin{proof}
	Note that this is the case in which the set of red vertices coincides with $V$. Also note that we only need to show how to change the color of each blue vertex in the path $P_\Phi$ built in the previous proof, without interfering with any possible augmentation of the red vertices of $P_\Phi$.

	Consider an edge $(a,b) \in P_\Phi$, with $a$ and $b$ being blue vertices. We replace the edge $(a,b)$ by the gadget shown in Figure~\ref{fig:path_gadget_red}, which is a path connecting $a$ and $b$, so that the colors of $a$ and $b$ are changed to red. In this gadget, we identify two sets of red vertices,
the 22 cross-shaped vertices and the 22 disk-shaped vertices. Observe that the disk-shaped vertices are placed in such a way that they are either adjacent or they are not visible to each other. Therefore, it is not possible to join two disk-shaped vertices with an edge. Observe as well that cross-shaped vertices are the only vertices of the graph that are visible from disk-shaped vertices. This implies that the cross-shaped vertices are forced to join the disk-shaped vertices in any matching, because otherwise there would be red disk-shaped vertices not meeting their parity constraints.

	Consider now an edge $(a,b) \in P_\Phi$, with $a$ being red and $b$ being blue. We can subdivide $(a,b)$ into two edges $(a,a')$ and $(a',b)$ and mark $a'$ as blue. Then, we apply the previous gadget to $(a',b)$. By repeating this operation for all blue edges and all red-blue edges of $P_{\Phi}$, we obtain a path $P'$ with all its vertices marked as red. The proof follows since, if we are able to augment $P'$ with a plane geometric perfect matching, then we are able to find a plane geometric perfect matching between the red vertices of $P_{\Phi}$, and viceversa.
\end{proof}

To finish this section, we remark that deciding if a plane geometric path can be augmented to a Eulerian plane geometric graph by adding a matching between its odd degree vertices can be clearly solved in linear time. The only vertices of a path with odd degree are its endpoints, and checking if the segment connecting these two endpoints crosses the path can be done in linear time.

\section{Plane topological augmentation in MOPs}\label{sec:mops}

To the best of our knowledge, the family of plane topological trees is so far the only family of plane topological graphs, for which the plane topological augmentation problem to meet parity constraints is polynomial~\cite{aldana2017planarity}. In this section we show that this problem is also polynomial for the family of maximal outerplane graphs.

An \emph{outerplanar graph} $G$ is a graph that has a plane embedding where all of its vertices belong to the unbounded face of the embedding. We shall refer to these embeddings as \emph{outerplane graphs}. A graph $G$ is a maximal outerplanar graph if it is not possible to add any edge to it, such that the resulting graph is still outerplanar.
A \emph{maximal outerplane graph} (for short a MOP) is an outerplane graph of a maximal outerplanar graph.

For the sake of clarity, in the rest of this section we will assume that all the vertices of a MOP $G$ are placed on an unit circle $\mathcal C$, and that they are labeled from $v_1$ to $v_{n}$ in clockwise order. We also assume that all of the edges of $G$ are \emph{arcs} of $\mathcal C$ connecting consecutive vertices, or straight-line segments contained in the interior of $\mathcal C$ called \emph{diagonals}. Therefore, $G$ consists of $n$ arcs and $n-3$ diagonals. We denote by $[v_i,v_j]$ the set of vertices $\{v_i, v_{i+1}, \ldots, v_{j} \}$ (mod $n$), and we call $[v_i,v_j]$ an interval of $G$.

Given a MOP $G$ and a set $C_G$ of parity constraints, recall that a vertex $v\in G$ is red if its parity must change to meet its parity constraint (that is, $v \in R(C_G)$), and blue otherwise. Further, recall that an edge of $G$ is red (respectively blue), if both of its endpoints are red vertices (respectively blue). An edge of $G$ is red-blue if its endpoints have different colors. Consider two red-blue diagonals of $G$. We say that they are \emph{parallel} if there is a straight line that separates their red endpoints from their blue endpoints. For instance, the red-blue diagonals $(v_m,v_i)$ and $(v_l,v_k)$ in Figure~\ref{fig:diag_rb-j.rb} are parallel, and the red-blue diagonals $(v_l,v_i)$ and $(v_j,v_k)$ in Figure~\ref{fig:diag_rb-br} are not. Note that the cyclic order of the colors of the four endpoints is red, red, blue, blue if the two red-blue diagonals are parallel, and red, blue, red, blue if they are not.

In the following theorem we characterize the MOPs that are topologically augmentable.

\begin{theorem} \label{main1}
	Let $G=(V,E)$ be a MOP and let $C_G$ be a set of parity constraints. Then, $G$ is topologically augmentable to meet $C_G$ if one of the following conditions is satisfied:
\begin{enumerate}\label{teo:car_mops}
		\item[(i)] $G$ contains a blue diagonal.
		\item[(ii)] $G$ contains two non-parallel red-blue digonals.
		\item[(iii)] There exist two parallel red-blue diagonals $(v_i,v_m), (v_k,v_l) \in E$ such that $v_i, v_k$ are blue, $v_l, v_m$ are red (possibly with $v_l=v_m$), and there exists a vertex $v_j$ of degree two such that the order of these vertices along $\mathcal{C}$ is $v_i,v_j,v_k,v_l, v_m$ in clockwise direction.
\end{enumerate}
Otherwise, $G$ is non-augmentable.	
\end{theorem}

\begin{proof} We first prove that if $G$ satisfies at least one of (i), (ii) or (iii), then it is topologically augmentable.

Case (i): Suppose that $G$ contains a blue diagonal $(v_i,v_j)$. Then, each of the intervals $A=[v_{j+1},v_{i-1}]$ and $B=[v_{i+1},v_{j-1}]$ contains at least one vertex. If one of these intervals, say $A$, contains only blue vertices, then we join a blue vertex in $A$ to all of the evenly many red vertices in $B$, and we are done. Suppose then that each of $A$ and $B$ contains at least one red vertex. In the clockwise traversal of $\mathcal{C}$, let $v\in A$ and $v'\in B$ be the red vertices that precede $v_i$ and $v_j$, respectively. Note that these vertices are not necessarily $v_{i-1}$ and $v_{j-1}$. One can easily verify that joining $v$ to all red vertices in $B$ (except possibly $v'$) and joining $v'$ to all red vertices in $A$ (except possibly $v$), every red vertex will meet its parity constraint in the resulting graph. See Figure~\ref{fig:diag_b-b}.

	\begin{figure}[ht!]
	\centering
		\begin{subfigure}[t]{0.32\textwidth}
			\includegraphics[width=0.9\linewidth]{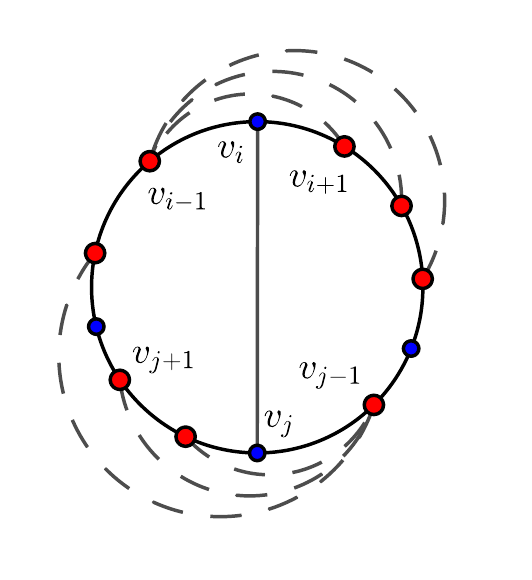}
			\captionsetup{width=.95\linewidth}
			\caption{}
			\label{fig:diag_b-b}
		\end{subfigure}~~~%
		\begin{subfigure}[t]{0.32\textwidth}
			\includegraphics[width=0.9\linewidth]{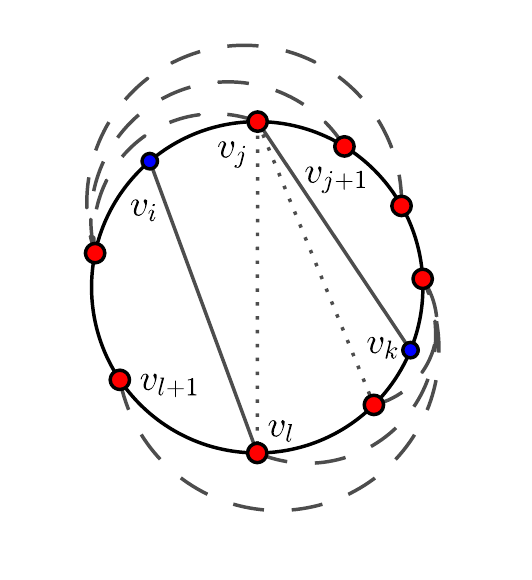}
			\captionsetup{width=.95\linewidth}
			\caption{}
			\label{fig:diag_rb-br}
		\end{subfigure}~~~%
		\begin{subfigure}[t]{0.32\textwidth}
			\includegraphics[width=0.9\linewidth]{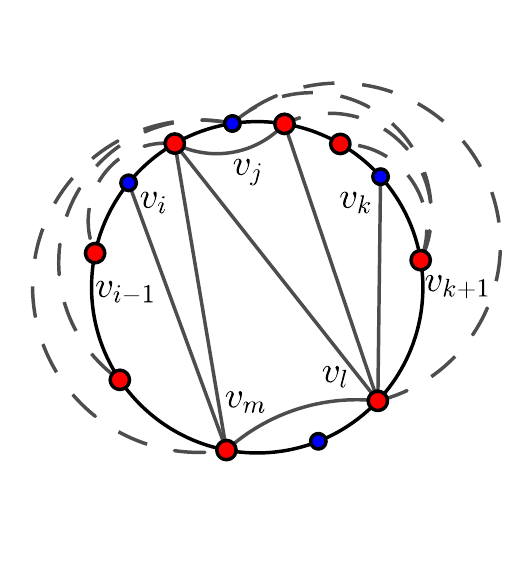}
			\captionsetup{width=.95\linewidth}
			\caption{}
			\label{fig:diag_rb-j.rb}
		\end{subfigure}%

		\caption{Topologically augmentable MOPs. (a) A blue diagonal exists. (b) Two non-parallel red-blue diagonals exist. (c) Two parallel red-blue diagonals and a degree-2 vertex exist.}
	\end{figure}

Case (ii): Suppose that $G$ contains two non-parallel edges $(v_i,v_l)$ and $(v_j, v_k)$ such that $v_i$ and $v_k$ are blue vertices, and $v_j$ and $v_l$ are red vertices. Assume without loss of generality that the relative order of these vertices along $\mathcal{C}$ is $v_i,v_j,v_k,v_l$, see Figure~\ref{fig:diag_rb-br}. In this case, we can proceed as in Case (i), as if the blue edge $(v_i,v_k)$ was an edge of $G$.

Case (iii): We can connect $v_{i-1}$ to each red vertex in $[v_{i+1},v_{j-1}]$ and  $v_{k+1}$ to each red vertex in $[v_{j+1},v_{k-1}]$, see Figure~\ref{fig:diag_rb-j.rb}. Note that after adding these edges, $v_{i-1}$ and $v_{k+1}$ could meet their parities or not.  Then, we connect all red vertices in $[v_{k+1},v_{i-1}]$ to $v_j$, including $v_{i-1}$ and $v_{k+1}$ if they do not meet their parities after the previous addition of edges. Observe that it does not matter what color $v_j$ has, the parities of all vertices in $G$ are met (see Figure~\ref{fig:diag_rb-j.rb}).

We now prove that if $G$ is topologically augmentable, then it must satisfy at least one of the three previous conditions.

Suppose that $G$ satisfies neither (i), nor (ii), nor (iii). Then, $G$ must satisfy the following conditions.
	
\begin{enumerate}
	
\item	By (i), $G$ can have only red or red-blue diagonals.
\item	By (ii), $G$ cannot have non-parallel red-blue diagonals. This implies that there is an interval $U$ in $\mathcal{C}$ containing all of the blue endpoints of all red-blue diagonals of $G$, and none of the red endpoints of these diagonals. Assume that $U$ is the shortest such interval and that, relabeling the vertices if necessary, $U=[v_1,v_i]$. See Figure~\ref{fig:mop_no_augmentable-rb}. Observe that if $G$ does not have red-blue diagonals, then $U$ is empty, and if $G$ contains only one red-blue diagonal, then $U$ consists of one blue vertex. Also observe that  $D=[v_{i+1},v_{n}]$ must always contain red vertices, regardless if $U$ is empty or not.

We claim that if $D$ contains a blue vertex $v_r$, then it must have degree two in $G$. To prove this claim, we distinguish whether $U$ is empty or not. Assume that $U$ is non-empty and suppose that there is a diagonal $(v_r,v_s)$ incident to $v_r$ in $G$. Thus, $v_s$ must be a red vertex, because otherwise a blue diagonal would exist and $G$ would satisfy (i). In addition, $v_s$ cannot lie between $v_{i+1}$ and $v_r$ clockwise, for otherwise $(v_r, v_s)$ is a red-blue diagonal and thus $v_r$ would belong to $U$. Finally, if $v_s$ lies between $v_r$ and $v_1$ clockwise, then non-parallel blue-red diagonals would exist, and thus $G$ would satisfy (ii). Hence, $v_r$ must have degree two in $G$. Assume now that $U$ is empty, so there are no red-blue diagonals. Since $G$ does not contain either blue diagonals, then no diagonal can be incident to $v_r$ and the claim follows.

As a consequence of this claim, there are no consecutive blue vertices in $D$, that is, if $v_r$ is a blue vertex in $D$ and $v_{r-1}$ and $v_{r+1}$ belong to $D$, then $v_{r-1}$ and $v_{r+1}$ are red vertices, and they are adjacent in $G$.

	\begin{figure}[ht!]
		\centering	
		\begin{subfigure}[t]{0.32\textwidth}
			\includegraphics[width=0.9\linewidth]{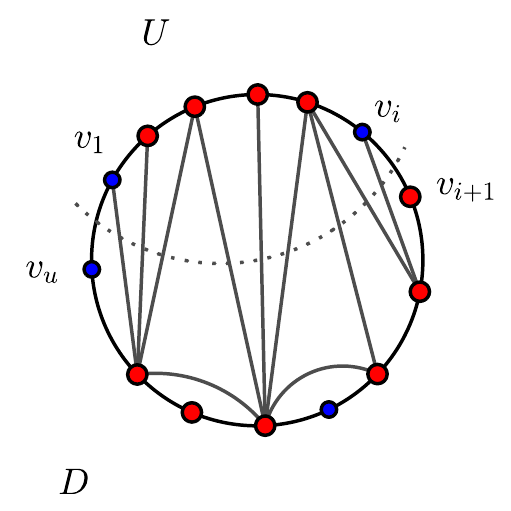}
			\captionsetup{width=.95\linewidth}
			\caption{}
			\label{fig:mop_no_augmentable-rb}
		\end{subfigure}%
		\begin{subfigure}[t]{0.32\textwidth}
			\includegraphics[width=0.9\linewidth]{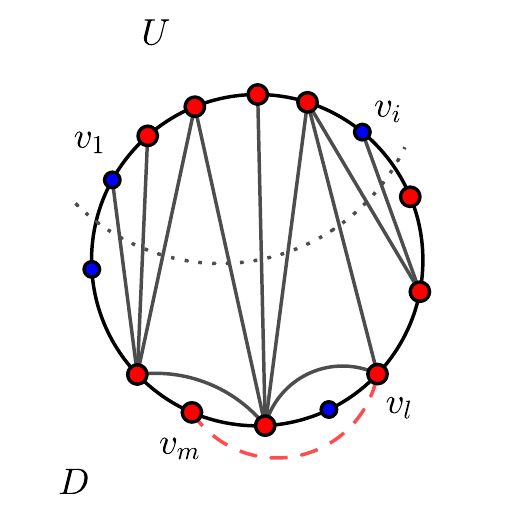}
			\captionsetup{width=.95\linewidth}
			\caption{}
			\label{fig:mop_no_augmentable2-rb}
		\end{subfigure}%
		\begin{subfigure}[t]{0.32\textwidth}
			\includegraphics[width=0.9\linewidth]{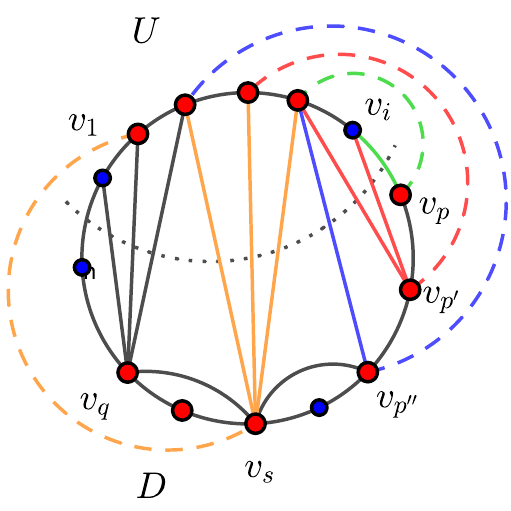}
			\captionsetup{width=.95\linewidth}
			\caption{}
			\label{fig:mop_no_augmentable3-rb}
		\end{subfigure}%
		\caption{(a) A non-augmentable MOP with blue-red diagonals. (b) Isolating a red vertex when two vertices in $D$ are connected. (c) Connecting red vertices from $D$ to a vertex in $U$.}
	\end{figure}

\item	Since $G$ does not satisfy (iii), all vertices in $U$ have degree at least three. Therefore, there is no diagonal connecting two vertices in $U$, because otherwise a vertex of degree 2 would exist in $U$. See Figure~\ref{fig:mop_no_augmentable-rb}.
\end{enumerate}

	Suppose that $G$ is augmentable to meet $C_G$ and let $H=(V,E')$ be a plane topological graph that augments $G$. We show that $H$ cannot exist if $G$ does not satisfy (i), (ii), or (iii).

We first prove that there is no edge in $H$ connecting two vertices in $D$. Assume to the contrary that such edges exist. Among all of them, we take the edge $e=(v_l,v_m) \in E'$ such that the interval $[v_l,v_m]$ belongs to $D$ and is the shortest one. Note that $v_l$ and $v_m$ can neither be consecutive vertices, nor be at distance two if there is a vertex of degree two between them, since joining them with an edge would cause duplicated edges in $G \cup H$. See Figure~\ref{fig:mop_no_augmentable2-rb}. Thus, $[v_l,v_m]$ contains at least one red vertex different from $v_l$ and $v_m$, since any blue vertex in $D$ must have degree two in $G$ and there are no two consecutive blue vertices in $D$. By the minimality of $[v_l,v_m]$, one can easily verify that the parities of the red vertices in $[v_l,v_m]$ cannot all be met. Therefore, there is no edge in $H$ connecting two vertices in $D$. In particular, if $U$ is empty, this result implies that $H$ cannot exist, since there is no way of meeting all parity constraints by adding edges incident to vertices in $D$.

Hence, we assume that $U$ is non-empty in the rest of the proof, so $G$ does contain parallel red-blue diagonals. Let $D_R=\{v_{p}, v_{p'}, ... , v_{q}\}$ be the set of red vertices in $D$ that are adjacent in $G$ to at least one vertex in $U$, where $p<p'<\ldots <q$. Note that $D_R$ is non-empty because there are red-blue diagonals. Consider a vertex $v_s \in D_R$. By the third condition, there is no diagonal in $G$ connecting two vertices in $U$, so the set of vertices in $U$ adjacent to $v_s$ forms an interval $[v_{l_s},v_{r_s}]$ along $\mathcal{C}$, $1 \leq l_s \leq r_s \leq i$, see Figure~\ref{fig:mop_no_augmentable3-rb}. Since $v_s$ is red and there are no edges in $H$ connecting two vertices in $D$, then $H$ must contain at least one edge connecting $v_s$ to a vertex $w$ in $U$, in order to meet the parity constraint of $v_s$. In addition, $w$ cannot be in $[v_{l_s},v_{r_s}]$ as any vertex in this interval is already adjacent to $v_s$.

Consider now $v_p$ and the interval $[v_{l_p},v_{r_p}]$. Note that $v_{r_p}=v_i$. In order to meet its parity constraint, $v_p$ must be adjacent to a vertex $v_{t_p}$ with $t_p < l_p$, i.e. $v_{t_p}$ lies to the left of $v_{l_p}$. In a recursive way, we can prove that any $v_s$ in $D_R$ must be adjacent to a vertex $v_{t_s}$ in $U$ to the left of the interval $[v_{l_s},v_{r_s}]$. This also applies to $v_q$, but the neighbours of $v_q$ in $U$ form an interval that starts at $v_1$, so $v_{t_q}$ must be to the left of $v_1$ and thus it cannot be in $U$, which is a contradiction. It follows that $H$ cannot exist and if $G$ is topologically augmentable, then it must satisfy at least one of (i), (ii), or (iii).
\end{proof}

It is clear that given a MOP $G$, we can check in linear time if each one of (i),(ii), and (iii) is satisfied. Thus, we can decide if $G$ is topologically augmentable to meet a set of parity constraints in $\mathcal{O}(n)$ time. In addition, it is easy to see that the set of edges added to $G$ in the proof of Theorem~\ref{main1} is not necessarily one with minimum size. We next provide a polynomial time algorithm to compute a minimum set of edges to augment a MOP $G$ (if it exists) to meet a set of parity constraints.

\begin{theorem} 	\label{teo:comp_min_set_mops}
Let $G$ be a MOP with $n$ vertices, and let $C_G$ be a set of parity constraints. It takes $\mathcal{O}(n^3)$ time to find a plane topological graph $H$ of minimum size such that $G$ is topologically augmentable by $H$ to meet $C_G$.
\end{theorem}

\begin{proof}
We compute an optimum edge set, if it exits, by dynamic programming. Given an interval $[v_i,v_j]$, we define $C[i,j]$ as the size of a smallest set of edges having both endpoints in $[v_i,v_j]$ such that this set is a plane topological graph that does not cross $G$ and meets the parity constraints of all vertices in $[v_i,v_j]$ (or $\infty$ if such a set of edges does not exist). If $\overline{v}$ denotes the color exchange of $v$ for any vertex $v$, we analogously define $C[\overline{i},j]$, $C[i,\overline{j}]$ and $C[\overline{i},\overline{j}]$ as the sizes of the optimum solutions as described previously, for the intervals $[\overline{v_i},v_j]$, $[v_i,\overline{v_j}]$ and $[\overline{v_i},\overline{v_j}]$, respectively.

We also define $D[i,j]$, $D[\overline{i},j]$, $D[i,\overline{j}]$ and $D[\overline{i},\overline{j}]$ as the size of an optimum solution for the intervals $[v_i,v_j]$, $[\overline{v_i},v_j]$, $[v_i,\overline{v_j}]$ and $[\overline{v_i},\overline{v_j}]$, respectively, assuming that the edge $(v_i,v_j)$ belongs to the solution. We next show how to compute $C[i,j]$ and $D[i,j]$. The table entries $C[\overline{i},j]$, $C[i,\overline{j}]$, $C[\overline{i},\overline{j}]$, $D[\overline{i},j]$, $D[i,\overline{j}]$ and $D[\overline{i},\overline{j}]$ are computed in a similar manner.

To compute $C[i,j]$, we look at the last vertex $v_k \in [v_i,v_j]$ in clockwise order (if it exists) connected to $v_i$ in an optimum solution for the interval $[v_i,v_j]$. If $v_k$ does not exist because $v_i$ is not connected to any other vertex of the interval, then $C[i,j] = C[i+1,j]$ if $v_i$ is blue, or $C[i,j] = \infty $ if $v_i$ is red. If $v_k$ exists, then the optimum solution can be computed by solving two subproblems. On the one hand, finding the minimum solution for the interval $[v_i,v_k]$ assuming that the edge $(v_i,v_k)$ belongs to the solution, and on the other hand, finding the minimum solution  for the interval $[v_k, v_j]$. Observe that, if $v_k$ is blue, then we can decide either to keep $v_k$ as blue in both subproblems or change its color in both subproblems, and if $v_k$ is red, then we can choose $v_k$ as red in one of the subproblems and blue in the other one. We do this analysis only if the edge $(v_i,v_k)$ does not belong to $G$. Summarizing, we have:

\begin{small}
\begin{align}
		C[i,j] &= \min
		\begin{cases}
            \begin{cases}
			C[i+1,j]	& \text{if $v_i$ is blue}\\
			\infty	& \text{if $v_i$ is red}
            \end{cases}\\\\
			\min_{i+2\le k\leq j}
            \begin{cases}
            \infty & \text{if $(v_i,v_k)\in E$}\\
			\min\{D[i,k] + C[k,j], D[i,\overline{k}] + C[\overline{k},j]\}	& \text{if $v_k$ is blue and $(v_i,v_k)\not\in E$}\\
			\min\{D[i,k] + C[\overline{k},j], D[i,\overline{k}] + C[k,j]\}	& \text{if $v_k$ is red and $(v_i,v_k)\not\in E$}\\
            \end{cases}
		\end{cases}
		\label{eq:min_cost_c}
\end{align}
\end{small}

The analysis to compute $D[i,j]$ is similar. Assuming that the edge $(v_i,v_j)$ belongs to an optimum solution, either $v_i$ is not connected to any other vertex of the interval $[v_i,v_j]$ except for $v_j$, or there exists a vertex $v_k$ (different from $v_j$) that is the last vertex of the interval connected to $v_i$. An optimum solution can be computed by finding an optimum solution for the interval $[v_i,v_k]$ assuming that $(v_i,v_k)$ belongs to the solution, finding an optimum solution for the interval $[v_k,v_j]$, and adding the edge $(v_i,v_j)$.

\begin{small}
\begin{align}
		D[i,j] &= \min
		\begin{cases}
            \begin{cases}
			C[i+1,\overline{j}]+1	& \text{if $v_i$ is red}\\
			\infty	& \text{if $v_i$ is blue}
			\end{cases}\\\\
            \min_{i+2\le k < j}
            \begin{cases}
            \infty & \text{if $(v_i,v_k)\in E$}\\
			\min\{D[\overline{i},k] + C[k,\overline{j}], D[\overline{i},\overline{k}] + C[\overline{k},\overline{j}]\}+1	& \text{if $v_k$ is blue and $(v_i,v_k)\not\in E$}\\
			\min\{D[\overline{i},k] + C[\overline{k},\overline{j}], D[\overline{i},\overline{k}] + C[k,\overline{j}]\}+1	& \text{if $v_k$ is red and $(v_i,v_k)\not\in E$}\\
            \end{cases}
		\end{cases}
		\label{eq:min_cost_d}
\end{align}
\end{small}

Equations~(\ref{eq:onevertex}) and~(\ref{eq:twovertices}) show the base cases for $C[i,j]$, when the intervals consist of one or two vertices, and Equation~(\ref{eq:threevertices}) shows the base case for $D[i,j]$, when the intervals consist of three vertices. 	

	\begin{align}
		C[i,i] &=
		\begin{cases}\label{eq:onevertex}
			0        & \text{if $v_i$ is blue}\\
			\infty   & \text{if $v_i$ is red}
		\end{cases}
		\\
		C[i,i+1] &=
		\begin{cases}\label{eq:twovertices}
			0        & \text{if $v_i$ and $v_{i+1}$ are blue}\\
			\infty   & \text{otherwise}
		\end{cases}
		\\
		D[i,i+2] &=
		\begin{cases}\label{eq:threevertices}
			1        & \text{if $v_i$ and $v_{i+2}$ are red and $v_{i+1}$ is blue}\\
			\infty   & \text{otherwise}\\
		\end{cases}
	\end{align}

The algorithm computes all tables in increasing order of the size of the intervals. Moreover, an extra table can be filled it out, as long as the recurrence is carried on, to store the optimum edge set of each interval. Note that $C[1,n]$ stores the size of the optimum solution for $G$, if it exists. The algorithm runs in $\mathcal{O}(n^3)$ time since there are $\mathcal{O}(n^2)$ intervals to explore and $v_k$ is computed in linear time for each interval.
\end{proof}

	\begin{figure}[ht!]
		\centering
		\begin{subfigure}[t]{0.30\textwidth}
			\includegraphics[width=\linewidth]{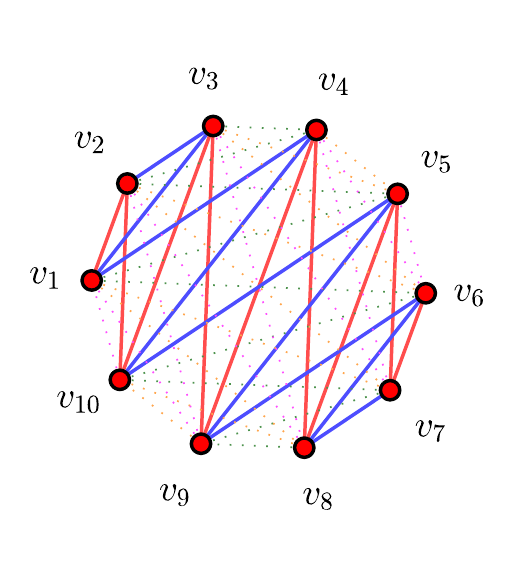}
			\captionsetup{width=.95\linewidth}
			\caption{}
			\label{fig:zig_paths}
		\end{subfigure}%
		~	
		\begin{subfigure}[t]{0.30\textwidth}
			\includegraphics[width=\linewidth]{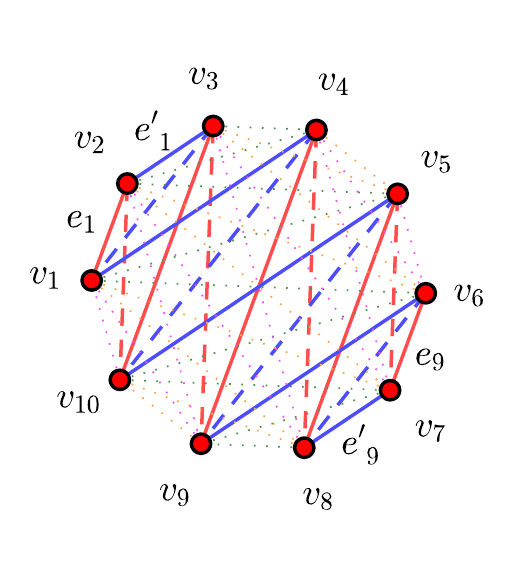}
			\captionsetup{width=.95\linewidth}
			\caption{}
			\label{fig:zig_paths2}
		\end{subfigure}%
		\caption{Decomposition of a complete graph into $n/2$ zig-zag paths. (a) $P_1$ in red, and $P_2$ in blue. (b) Two maximal matchings on each zig-zag path.}
	\end{figure}

It is not hard to see that the previous algorithm also works for (not necessarily maximal) outerplane graphs  when the new edges are added in the unbounded face. Notice that when we add exactly $|R(C_G)|/2$ edges to a graph $G$, they form a plane perfect matching between the red vertices of $G$.

While not all MOP's are topologically augmentable, it is straightforward to see that given a MOP $G$ and a set of parity constraints $C_G$, it is always possible to add edges to $G$ such that all but at most two parity constraints are satisfied. This can be done as follows: Since $G$ is a MOP, it always has a vertex $v_i$ of degree two. Thus, the edge $(v_{i-1},v_{i+1})$ belongs to $G$. If we join $v_i$ to all red vertices in $G$ to which it is not adjacent, then all but at most two parity constraints are satisfied. When we insist  to do the augmentation with a plane topological matching, then we can guarantee that the number of vertices not meeting their parity constraints is at most four.

\begin{theorem}
	Let $G=(V,E)$ be a MOP, and let $C_G$ be a set of parity constraints. Then, there exists a plane topological matching $M$ such that $G$ and $M$ are compatible and edge-disjoint, and $G\cup M$ meets all but at most four parity constraints.
\end{theorem}
\begin{proof}
Suppose that $R(C_G)=V$. The geometric complete graph $K_n$, with $n$ vertices placed in convex position, can be decomposed into $n/2$ zig-zag paths $P_1, ..., P_{n/2}$ as follows: $P_1=\{v_1, v_2, v_{n}, v_3,$ $v_{n-1}, ... \}$, $P_2=\{v_2,v_3,v_1,v_4,v_n, ...\}$, ..., $P_{n/2}=\{v_{n/2},v_{n/2+1},v_{n/2-1},...\}$. That is to say, $P_{i+1}$ is obtained by rotating the vertices of $P_i$ one position to the right, as shown in Figure \ref{fig:zig_paths}.

	We number the edges of each zig-zag path $P_i$ from 1 to $n-1$, starting at the edge $(v_i,v_{i+1})$. Observe that each path contains two maximal matchings, one of size $\frac{n}{2}$, given by the odd numbered edges (solid segments in Figure~\ref{fig:zig_paths2}), and the second of size $\frac{n}{2}-1$, given by the even numbered edges (dashed segments in Figure \ref{fig:zig_paths2}). In total, $K_n$ has $n$ pairwise disjoint maximal matchings.

	Since $G$ has $n-3$ diagonals, there exists at least one matching, say $M_i$, such that none of its edges is a diagonal of $G$. If $M_i$ is formed by odd (solid) edges, then there exist two edges on the boundary of the unbounded face of $G$ in common with $M_i$, so $M_i\setminus E$ spans $n-4$ vertices of $G$. If $M_i$ is formed by even (dashed) edges, then $M_i\cap E = \emptyset$ and $M_i$ spans $n-2$ vertices of $G$ since $|M_i|=\frac{n}{2}-1$. Notice that the edges of $M_i$ can be drawn in the unbounded face of $G$ without crossings, hence the theorem follows when $R(C_G)=V$.
	
	It is straightforward to see that when $R(C_G)\subset V$, we can apply the same analysis on the subgraph induced by the red vertices to prove the theorem.
\end{proof}

\section{Conclusions}\label{sec:conclusions}

In this paper we have addressed the plane topological augmentation problem to meet parity constraints and we have proved that deciding if a plane topological graph is topologically augmentable is computationally hard, under several assumptions on the input graph, the vertices that must change their parities, and the size of the augmentation. We have obtained analogous results for the plane geometric augmentation problem to meet parity constraints. We have also proved the hardness of some of these variants for plane geometric trees and paths. Table~\ref{tabla2} summarizes these results.

For the family of MOPs, we have characterized the MOPs that are topologically augmentable to meet parity constraints, and we have given an $\mathcal{O}(n^3)$ time algorithm to augment them with the minimum number of edges, if that is the case. We have also shown that there is always a topological augmentation by a plane topological matching, satisfying all but at most four parity constraints.

Finally, we conclude this paper with the following conjecture about the complexity of deciding if a plane geometric tree is geometrically augmentable.

\begin{conjecture}
	Let $T=(V,E)$ be a plane geometric tree and let $C_T$ be a set of parity constraints. Then, the problem of deciding if $T$ is geometrically augmentable is $\mathcal{NP}$-complete.
\end{conjecture}

\section*{Acknowledgments}

Alfredo Garc\'\i a, Javier Tejel and Jorge Urrutia were supported by H2020-MSCA-RISE project 734922 - CONNECT. Research of Alfredo Garc\'\i a and Javier Tejel was also supported by project MTM2015-63791-R (MINECO/FEDER) and by project Gobierno de Arag\'on E41-17R (FEDER). Research of Jorge Urrutia was also supported by UNAM project PAPIITIN102117. Research of Juan Carlos Catana was supported by Consejo Nacional de Ciencia y Tecnolog\'\i a, Mexico.


\bibliographystyle{abbrv}
\bibliography{refs}

\end{document}